\documentclass[10pt,journal,cspaper,compsoc]{IEEEtran}
\usepackage{indentfirst}
\usepackage{xspace}
\usepackage{amsmath,epsfig}
\usepackage{times}
\usepackage{soul}
\usepackage{pifont}
\usepackage{algorithmicx,algpseudocode}

\usepackage{subfigure}

\usepackage{algpseudocode}

\newtheorem{lemma}{Lemma}
\newtheorem{proof}{Proof}
\newtheorem{property}{Property}
\renewcommand{\algorithmicrequire}{\textbf{Input:}}
\renewcommand{\algorithmicensure}{\textbf{Output:}}
\usepackage[normalem]{ulem}
\usepackage{enumerate}
\usepackage[justification=justified]{caption}

\usepackage{amsfonts}
\usepackage{tikz}

\usepackage[linesnumbered, ruled, vlined]{algorithm2e}
\usepackage{algpseudocode}

\renewcommand{\algorithmicrequire}{\textbf{Input:}}
\renewcommand{\algorithmicensure}{\textbf{Output:}}
\begin{document}
\bstctlcite{IEEEexample:BSTcontrol}

%\title{Fast Scans by Weaving Secondary Index into Data Storage in Main-memory Column-stores}
\title{ByteStore: Hybrid Layouts for Main-Memory Column Stores}

%\author{Wenjian Xu,
%		Eric Lo,
%		Yu Li,
%		Pengfei Zhang
%\IEEEcompsocitemizethanks{\IEEEcompsocthanksitem W. Xu and Y. Li are with the Department of Computing, Hong Kong Polytechnic University, Hong Kong.\protect\\
%E-mail: \{cswxu, csyuli\}@comp.polyu.edu.hk.\protect \\
%E. Lo and Pengfei Zhang are with the Department of Computer Science and Engineering, Chinese University of Hong Kong, Hong Kong.\protect\\
%E-mail: \{ericlo, pfzhang\}@cse.cuhk.edu.hk.\protect 
%}}

\author{Pengfei~Zhang,
		Ziqiang~Feng,
		Eric~Lo,
		Hailin~Qin
\IEEEcompsocitemizethanks{\IEEEcompsocthanksitem 
P. Zhang, Eric Lo, Hailin are/was with the Department of Computer Science and Engineering, Chinese University of Hong Kong, Hong Kong. Z. Feng is with Google. }
}
% \IEEEcompsocthanksitem  H. Qin is with the University of California San Diego.  Work done while at CUHK. 
% E-mail: haqin@eng.ucsd.edu.
%\IEEEcompsocitemizethanks{This work is supported by Hong Kong General Research Fund (14200817, 15200715, 15204116), Hong Kong AoE/P-404/18, Innovation and Technology Fund ITS/310/18.}

% \IEEEcompsocitemizethanks{This work is supported by Hong Kong General Research Fund (14200817, 15200715, 15204116), Hong Kong AoE/P-404/18, Innovation and Technology Fund ITS/310/18.}
% }

%\markboth{IEEE Transactions on Knowledge and Data Engineering (TKDE)}{}%
%\IEEEtitleabstractindextext{
\IEEEcompsoctitleabstractindextext{
\begin{abstract}
The performance of main memory column stores highly depends on the scan
and lookup operations on the base column layouts. 
Existing column-stores adopt a homogeneous column layout, leading to sub-optimal performance on real
workloads since different columns possess different data characteristics. 
In this paper, we propose ByteStore, a column store that uses different storage layouts for 
different columns. 
We first present a novel data-conscious column layout,
PP-VBS (Prefix-Preserving Variable Byte Slice).
PP-VBS exploits data skew to accelerate scans 
without sacrificing lookup performance. 
Then, we present an experiment-driven column layout advisor to select individual column layouts 
for a workload. 
%Extensive experiment on real and synthetic data shows ByteStore 
%XXXXXXXXXXXXX (PFPFPF).
Extensive experiments on real data show that ByteStore outperforms homogeneous storage engines by up to 5.2$\times$.

\end{abstract}

}

\maketitle

\IEEEdisplaynotcompsoctitleabstractindextext

\IEEEpeerreviewmaketitle

\section{Introduction}\label{sec:introduction}

Main-memory column stores are popular for fast analytics of relational data \cite{farber2012sap,monetdbx100,abadi2009column}. By holding the data inside the memory of a server, or the 
aggregated memory of a cluster, these systems eliminate the disk I/O bottleneck and have 
the potential to  unleash the  high performance locked in modern CPU-memory stacks, 
including multiple cores, simultaneous multi-threading (SMT), single-instruction 
multiple-data (SIMD) instruction sets, hierarchical cache, and large DRAM bandwidth.

OLAP workloads are usually read-only.
% and column-oriented execution model is used to process queries to avoid loading
% irrelevant columns into the CPU cache.
To fully utilize the power of column-oriented storage, 
denormalization is often used to 
%have been proposed to 
transform tables into one or a few outer-joined wide tables
such that expensive joins and nested queries can then be flattened
as simple scan-based queries on the relevant columns~\cite{li2014widetable, byteslice, li2013bitweaving}.
Under this scan-heavy paradigm, most of the query time is spent on two operations that
directly consume the \emph{base columns}: scan and lookup.
The \emph{scan} operation on a column filters row IDs whose column values
satisfy a predicate (e.g., \texttt{year < 2018}). Given these row IDs, the \emph{lookup}
operation extracts column values into their plain form (e.g., \texttt{int32})
to be consumed by upstream operations, such as sorting~\cite{balkesen2013multi, xuSigmod16}, aggregation~\cite{shatdal1995adaptive}, 
and so on. %``LIKE'' predicate for strings. 
The overall performance of queries thus heavily depends on 
the scan and lookup performance 
on the base columns \cite{byteslice, li2013bitweaving}.

%{\bf \{PF should come back to add some materials here.\}}
Recently, there has been a flurry of research for in-memory \emph{base column layouts} such as Bit-Packed \cite{willhalm2009simd,simd-scan2}, PE-VBP \cite{li2013bitweaving,paddedEncoding} and ByteSlice \cite{byteslice}, which enable fast column scans.
%that enables fast column scans \cite{byteslice, li2013bitweaving, paddedEncoding, willhalm2009simd,simd-scan2}.
From the system perspective, 
however, all of them are column stores 
with a \emph{homogeneous} storage layout, i.e., the same storage layout is used across all data columns.
Figure~\ref{taxitrip_per_column} hints the drawback of such a homogeneous approach.
In the figure, it shows the scan performance on the
state-of-the-art main-memory column layouts: 
ByteSlice~\cite{byteslice} and PE-VBP (stands for Padded Encoding Vertical Bit Packing)~\cite{paddedEncoding} on a real dataset \cite{taxitrip} with 23 columns. 
%We give performance measurements in terms of processor cycles per column code. 
%Lower number of cycles per code means better scan performance.
We run the experiments using one core and give performance measurements in terms of processor cycles per column code.
It is clear that no one size fits all: 
PE-VBP, which is skew-aware,
outperforms ByteSlice, which is skew-agnostic, on 9 out of 23 columns.
Meanwhile, ByteSlice outperforms PE-VBP for the rest.
The above motivates us to design ByteStore, a new storage engine for column stores.
ByteStore adopts different storage layouts for different data columns.
ByteStore is different from other hybrid storage engines like HYRISE \cite{grund2010hyrise}.
HYRISE uses a hybrid of row/column storage for HTAP workloads.
By contrast, ByteStore is a pure column store that focuses on OLAP workload 
but it is hybrid in terms of using different encoding and layouts to store the columns.
To our best knowledge, this paper is the first to use different storage layouts for different base columns
in main-memory column stores.

\begin{figure}\centering
\includegraphics[width=1.0\columnwidth]{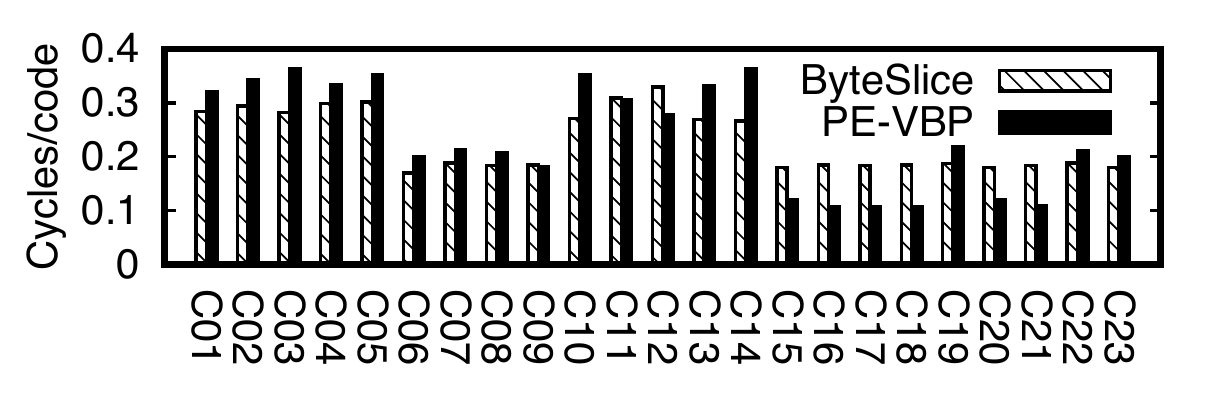}
\caption{Scan Speed on a Real Dataset: ByteSlice vs PE-VBP}
\label{taxitrip_per_column}
\vspace{-0.6cm}
\end{figure}

Technically, 
ByteStore is beyond a simple integration of ByteSlice and PE-VBP --- the two most efficient storage layouts to date for non-skewed and skewed data, respectively. 
First, ByteStore abandons PE-VBP but uses a new storage layout to complement with ByteSlice instead.
The problem of PE-VBP is that 
it actually has poor lookup performance
because of its way of storing the bits of a data value in
scattered memory locations.
Furthermore, 
PE-VBP indeed has limited performance edge over ByteSlice when scanning skewed data 
because its way of encoding a column is 
\emph {prefix-free}, making almost all values are 
encoded with a longer code length.
Therefore, one contribution of this paper is PP-VBS (stands for Prefix Preserving Variable Byte Slice), 
a new storage layout that uses byte as the storage unit
and a new \emph{prefix-preserving} method to encode the skewed columns.
Unlike ByteSlice, PP-VBS is skew aware.
Different from PE-VBP,
PP-VBS leverages data skewness 
to accelerate scans without hurting lookup performance.

With the advent of PP-VBS, we have observed that ByteSlice layout dominates scan and lookup performance on uniform to lightly skewed data columns and 
PP-VBS dominates the rest.
Therefore, ByteStore has to 
make a binary decision between the two storage layouts for a given data column.
Although the decision boundary looks complicated 
--- ByteSlice and PP-VBS outperform each other based on an array of factors such as the workload, 
value distribution and domain sizes,
we observe that there exists a simple yet reliable decision boundary based solely on data skewness.
Therefore, another contribution of this paper 
is an experiment-driven column-layout-advisor based on that observation.

Works on main-memory analytics mostly are evaluated on synthetic data (e.g., TPC-H) \cite{li2014widetable,byteslice, li2013bitweaving, xuSigmod16, paddedEncoding, imprint, columnSketches}.
Our last contribution is a comprehensive experimental study 
based on not only TPC-H but also six open datasets and workloads. 
Experiments show that ByteStore outperforms any homogeneous storage engines.
% (i.e., using only \emph{Bit-Packed} \cite{willhalm2009simd,simd-scan2}, using only ByteSlice or using only PE-VBP).
It therefore validates the effectiveness of having hybrid data layouts at column level.
%
%This paper is a journal extension of ByteSlice \cite{byteslice}.  
A preliminary version of this paper appears in \cite{byteslice}, in which only 
ByteSlice was discussed and evaluated.
In this version, we discusss
a full-fledged hybrid storage engine 
with ByteSlice as one of its components. 

The remainder of this paper proceeds as follows: Section \ref{sec:background} contains necessary background information;  Section \ref{sec:SkewScan} presents the new storage layout PP-VBS; Section \ref{section:cla} presents the column-layout advisor; Section \ref{sec:experiments} presents the experimental results. Section \ref{sec:related} discusses related works; Section \ref{sec:conclusion} gives the conclusion.

\section{Background and Preliminary}\label{sec:background}

\subsection{SIMD Instructions}
\emph{Data-level parallelism} is one strong level of parallelism supported by modern processors. 
Such parallelism is supported by SIMD (single instruction multiple data) instructions, which interact with $S$-bit SIMD registers as a \emph{vector of banks}. A bank is a continuous section of $b$ bits. 
In AVX2, $S=256$ and b is $8$, $16$, $32$ and $64$. 
We adopt these values in this paper since AVX2 is the most widely available in server processors (e.g., from Intel Haswell to TigerLake and AMD), but remark that our techniques can be straightforwardly extended to AVX-512 model.
The choice of $b$, the \emph{bank width}, is on per instruction basis. 
A SIMD instruction carries out the same operation on the vector of banks simultaneously. 
For example, the {\tt \_mm256\_add\_epi32()} instruction performs an 8-way addition between two SIMD registers, which adds eight pairs of 32-bit integers simultaneously. 
Similarly, the {\tt \_mm256\_add\_epi16()} instruction performs 16-way addition between two SIMD registers, which adds sixteen pairs of 16-bit short integers simultaneously. The degree of \emph{data-level} parallelism is $S/b$.
\subsection{Scan-based OLAP Framework}

Modern analytical column stores
transform complex queries into scan-heavy queries on denormalized wide
tables~\cite{li2014widetable,byteslice,li2013bitweaving}. These queries typically have extensive 
\texttt{WHERE} clauses requiring scan on many columns. 
A (column-scalar) scan takes as input a dictionary-encoded column, and a predicate of
types $=, \ne, >, <, \le, \ge$, {\tt BETWEEN}. The scalar literals in the predicates
(e.g., 2018 in \texttt{WHERE year < 2018}) are encoded using the same dictionary to 
encode the column values. By using order-preserving encoding, comparison on
codes yields correct result for comparison of the original column values. 
For predicates involving arithmetic or similarity search (e.g., {\tt LIKE} predicates on strings), codes have to be decoded before a scan is evaluated in the traditional way. 

The scan operation filters all matching column codes, and outputs a
\emph{result bit vector} to indicate the matching row IDs. The bit vector makes it
easy to combine scan results in logical expression (conjunction or disjunction),
and handle NULL values and three-valued Boolean logic~\cite{li2013bitweaving}.
All prior works on column scan~\cite{li2014widetable,byteslice,li2013bitweaving,paddedEncoding,willhalm2009simd,simd-scan2} support this bit vector interface,
so we follow the same convention in this paper.
After scan, column codes involved in projection, aggregation or sorting
%sub-string matching in ``LIKE'' conditions 
may need to be retrieved and reconstructed into 
their canonical forms  (e.g., \texttt{int32}). 
This is called lookup.
%Often, only a small fraction of values are looked up as a result of (selective) scan.
The lookup operation takes as input an encoded column and a result bit vector from
a scan. 
It then outputs the values as an array (e.g., \texttt{int32[]}).

Under this framework, scan and lookup are the two major operations 
whose performance directly depend on the base columns' layouts. 
Other operations such as sorting and aggregation are independent of the base columns.
\begin{figure}\centering
\includegraphics[width=1\columnwidth]{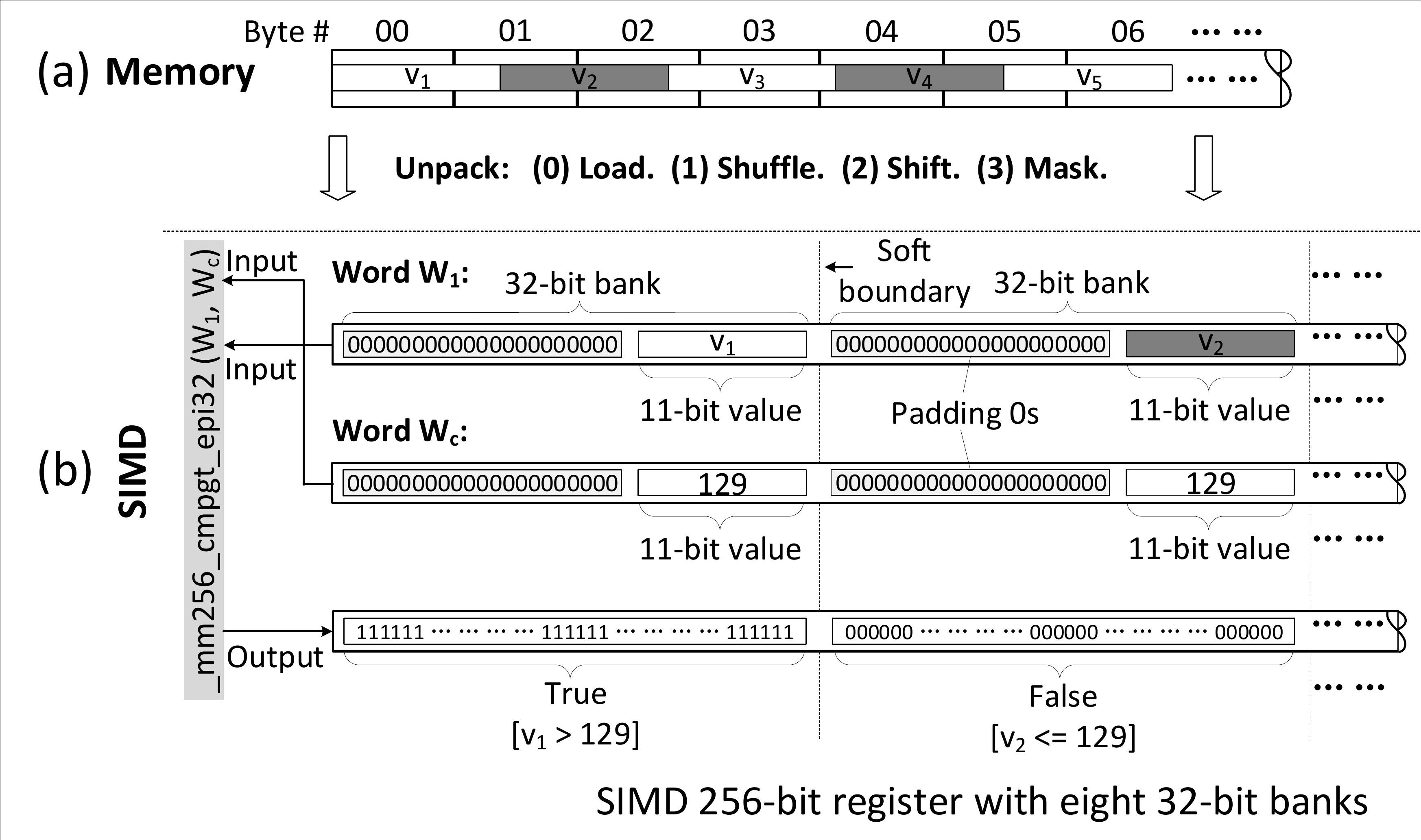}
%\vspace{-2ex}
\caption{(a) 11-bit column values stored under bit-packed layout (b) Scan on bit-packed data with predicate $v>129$}
\label{fig:bit-packed}
\vspace{-0.4cm}
\end{figure} 

\subsection{Bit-Packed}

The Bit-Packed layout \cite{willhalm2009simd,simd-scan2} aims to minimize the memory bandwidth usage when processing data. 
Same as ByteSlice, Bit-Packed is skew-agnostic and codes are fixed-length.
Figure \ref{fig:bit-packed}(a) shows an example with 11-bit column codes. Codes are tightly packed together in the memory, ignoring any byte boundaries. To be specific, the first code $v_1$ is put in the 1-st to 11-th bits whereas the second code $v_2$ is put in the 12-th to 22-nd bits and so on.
%In bit-packed layout, it is possible that a code spans multiple bytes (e.g., code $v_6$ spans Bytes\# 06, 07 and 08).

To evaluate a scan on a bit-packed column, it is necessary to unpack the tightly packed data into SIMD registers. 
Since a code may initially span 3 bytes (e.g., $v_3$) as shown in Figure \ref{fig:bit-packed}(a), each code has to be unpacked into a 32-bit bank of the SIMD register.  
Under AVX2 architecture (i.e., the length of SIMD registers is 256-bit), scan is run in 8-way (256/32) data level parallelism. 
In other words, eight 11-bit codes (e.g., $v_1\sim v_8$) are loaded from memory and aligned into eight 32-bit banks of the register. 
After unpacking, data in the SIMD register (e.g., $W_1$ in Figure \ref{fig:bit-packed}(b)) is ready to be processed by scan operation.
Figures \ref{fig:bit-packed}(b) shows how to evaluate a predicate $v > 129$ on the unpacked codes with AVX2's 8-way greater-than comparison instruction {\tt \_mm256\_cmpgt\_epi32()}. 
After that, the scan starts another iteration to unpack and compare the next 8 codes with $W_c$. 
In the example above, although 8-way parallelism is achieved in data processing, many cycles are actually wasted during unpacking. To align the 8 codes into the eight 32-bit banks, three extra SIMD instructions (i.e., \emph{Shuffle}, \emph{Shift} and \emph{Mask}) are carried out. 
Furthermore, as 0's are used to pad up with the SIMD soft boundaries, for above example, any data processing operation is wasting $(32-11)\times 8 = 168$ bits of computation power per cycle.

To retrieve a code $v_i$ from bit-packed layout, one has to gather all bytes that it spans.
For example, to look up $v_3$, Bytes\# 02$\sim$04 are fetched from the memory. 
%After that, %the fetched bytes are stitched together in order and we extract the code from it using bit shifting and masking.
%the corresponding bits of the code are extracted from each fetched byte and stitched together in order to generate the code.
As a code may span multiple bytes under the bit-packed format, retrieving one code may incur multiple cache misses, particularly when those bytes span across multiple cache lines.
% In real datasets, skewed columns are extensively common; in a skewed column, a small fraction of distinct values will have way higher frequency. 
% For example, in the three real-world datasets that are used in experiments, above 50\% of columns for each dataset are skewed except those 
% auto-generated ID, time-stamp, etc.
% %columns that could be easily distinguished as non-skewed columns such as 
% In the dataset of TAXITRIP \cite{taxitrip}, above 80\% of columns are skewed.  % except the columns of Taxi\_ID, Trip\_ID, Start\_TimeStamp and so on.

\subsection{PE-VBP}

%%- The PE paper said there are many data columns are skewed in nature. It is state-of-the-art work to exploit that to improve the performance of scan. In PE paper, \textbf{Zipfian distribution is applied to skewed columns}. 

%%- Describe Variable Length encoding, VBP and scan algorithm in details. Admire that variable length encoding is a very good idea to improve scan performance.

%In the paper of PE-VBP \cite{paddedEncoding}, column skew is discussed extensively and Zipf distribution is applied to skewed columns. Many other works, including \cite{knuth97, knuth98, christodoulakis1981}, 
%also used the Zipf distribution more generally to model distributions of values that are skewed.

\begin{figure}\centering
\includegraphics[width=0.8\columnwidth]{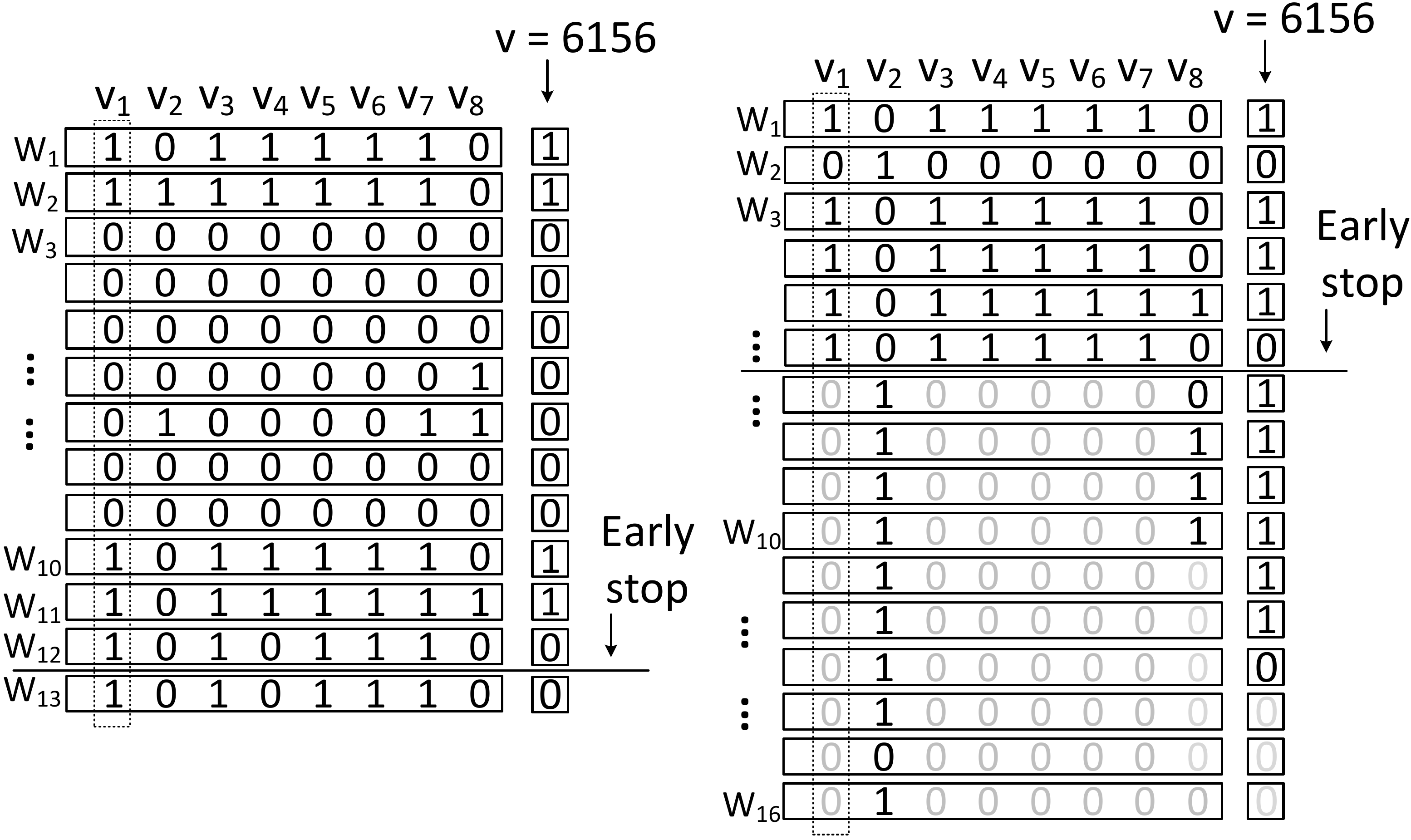}
\caption{(a) VBP and (b) PE-VBP}
\label{fig:pe-vbp}
\vspace{-0.4cm}
\end{figure}

Li et al.~\cite{paddedEncoding} proposed
a variable-length encoding scheme called \emph{Padded Encoding} (PE) that leverages 
data skew to accelerate scan. They combine PE with a vertical bit packing (VBP) storage layout~\cite{li2013bitweaving} to form PE-VBP, a skew-aware column scan technique.
Figure~\ref{fig:pe-vbp}(a) shows how VBP (without PE) stores a block of eight 13-bit column codes 
$v_1 \sim v_8$ in memory. Each horizontal box $W_i$ is a contiguous memory region.
The $j$-th bits of the codes are stored in $W_j$. In other words, the bits belonging 
to $v_1=(1100000001111)_2$ are vertically distributed across different memory regions (i.e., $W_1\sim W_{13}$).

The scan on VBP leverages a key insight that within each processor there is  abundant ``intra-cycle parallelism" as the processor's ALU operates on multiple bits 
in parallel. 
To illustrate, consider the evaluation of predicate ``{\tt v = 6156}" 
(see Figure \ref{fig:pe-vbp}(a)).
The scan algorithm compares the bits of $v_i$ with the bits of 6156 in a bit-by-bit
fashion, from the most significant bit to the least significant bit. Since the first
bits of $v_1 \sim v_8$ are stored together in $W_1$, they can be loaded into the CPU 
as a single memory word and processed in parallel using bitwise instructions
(in real implementation, a memory word is usually a SIMD register, e.g., 256 bits).
It compares the first bit of 6156 --- $(1)_2$ --- with all bits in 
$W_1 = (10111110)_2$. At this point, only the $v_2$ and $v_8$ are guaranteed to fail
the predicate because their first bit is 0. The other six codes are inconclusive.
Thus, the scan has to continue to the next (second) bit, and repeats.
After scanning the 12-th bits (i.e., $W_{12}$), %none of $v_i$'s first 12 bits match 6156's first 12 bits,
there are no codes whose first 12 bits match 6156's first 12 bits,
%--- $(110000000110)_2$ --- 
 safely declaring all codes in this block (i.e., $v_1\sim v_8$) fail the predicate.
There is no need to load $W_{13}$ into the CPU and process it.
In this case, the scan is said to \emph{stop early} and can proceed to the next
block of codes.

PE-VBP builds on VBP by leveraging 
data skew to increase the chance of early stop during a scan.\footnote{PE-VBP also considers skew in the predicate literals of a workload.
For example, if a predicate literal (e.g., the value 6156 in the predicate
``{\tt v = 6156}") is known and remains constant in all instantiations 
of a query template, PE-VBP will also encode those predicate-skewed values using fewer bits.  However, query literals rarely remain constant 
in real workloads.  
Therefore, we do not consider this type of skew in this paper.
}
Similar to Huffman encoding, PE assigns shorter codes for frequent values and longer 
codes for infrequent values. To store PE-encoded values in VBP, shorter codes are
padded with zeros to align with the longest code. Figure~\ref{fig:pe-vbp}(b) shows how
PE-VBP encodes and stores the same eight column values, where the grey bits are padding zeros.
It is worth noting that $v_1\sim v_8$ are encoded differently from VBP in Figure~\ref{fig:pe-vbp}(a).
The encoding scheme in PE-VBP is \emph{prefix-free}, which guarantees that if $a \ne b$, 
then the encoded value (called \emph{code}) of $a$ cannot be a prefix of that of $b$, or vice versa. This property
increases the maximum length of PE-encoded codes, for example, from 13 to 16 bits in 
Figure~\ref{fig:pe-vbp}.

Despite using more bits, PE-VBP achieves faster scan than VBP when the data is skewed.
Figure~\ref{fig:pe-vbp}(b) shows the scan can now stop early at the 6-th bits (i.e., $W_6$),
about $2\times$ faster than VBP,
because frequent values $v_1, v_3, v_4, v_5, v_6, v_7$ are encoded with only 6 bits,
as opposed to 13 bits in VBP. In general, PE-VBP increases the likelihood of early
stop after scanning each bit. Similar to many existing work~\cite{knuth97,paddedEncoding, christodoulakis1981,skew_index},  PE-VBP uses the Zipf distribution to model skewed data.

% We refer to this variable-length encoding (VLE) before padding zeros as the \emph{base encoding} of PE.
% %The prefix-free property is required for the base encoding to ensure that different codes can still be distinguished without ambiguity after padding zeros. An encoding scheme is prefix-free when it ensures one code would not be a prefix of another. For example, a scheme is not prefix-free if it gives values 'A', 'B' the codes $(11)_2$, $(110)_2$, beas
% In addition, the base encoding keeps prefix-free property. 
% An encoding scheme is prefix-free when it ensures one code would not be a prefix of another. 
% Figure \ref{fig:pe-vbp}(b) shows a PE-VBP block. 
% In this block, the 8 values are the same as that in the block of Figure \ref{fig:pe-vbp}(a). %, i.e., $v_1,...,v_8$. 
% The difference is that they are encoded as 8 16-bit codes with Padded Encoding (PE). 
% To illustrate, consider the evaluation of the same predicate ``v = 6156"(see Figure \ref{fig:pe-vbp}(b)), since PE assigns shorter codes for frequent values, which in turn increases the likelihood that codes have distinct values in the leading bits.
% In this example, the predicate evaluation can \textbf{stop early} after the bitwise comparison of $6$-th bit position. 
% As compared to the predicate evaluation in Figure \ref{fig:pe-vbp}(a), the scan on PE-VBP has the potential to be 2$\times$ faster.

Unfortunately, both VBP and PE-VBP suffer from very expensive lookup operation.
%When performing a lookup, 
When reconstructing a column code, both VBP and PE-VBP must retrieve every single bit of a code from a different memory word. 
Each bit is likely to reside in a different cache line,
and incur a cache miss, costing hundreds of CPU cycles. 
As shown in \cite{byteslice}, the expensive lookup often offsets 
the performance gain of fast scan from the whole query point of view.
The problem is exacerbated in PE-VBP 
since 
%it uses more bits per code on average. 
%Furthermore, 
the storage layout of PE-VBP is intrinsically at odds with its variable-length encoding scheme. 
As shown in Figure~\ref{fig:pe-vbp}(b), 
all short codes have to be padded up with zeros to align with the 
longest code. 
It retrieves 16 bits instead of 13 bits in VBP (see Figure \ref{fig:pe-vbp}(a)) to re-build a single code.

\subsection{ByteSlice}

\begin{figure}\centering
\includegraphics[width=0.9\columnwidth]{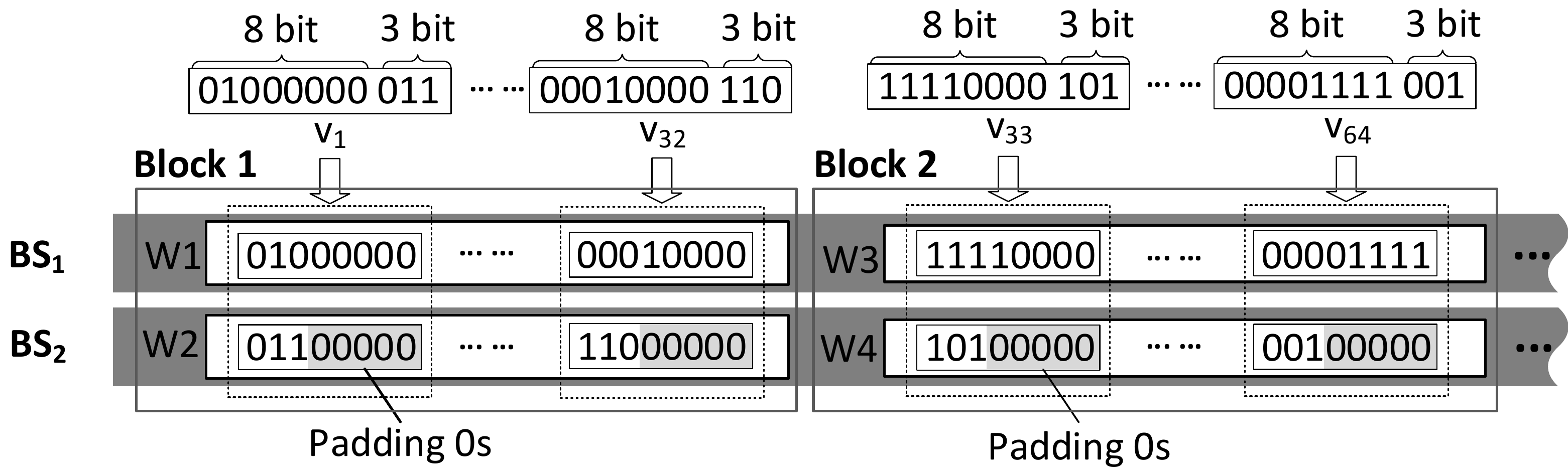}
\caption{ByteSlice}
\label{fig:byteslice}
\vspace{-0.4cm}
\end{figure}

ByteSlice is a state-of-the-art technique that achieves balanced high performance
of scan and lookup.
Figure~\ref{fig:byteslice} shows how ByteSlice~\cite{byteslice} arranges 64 11-bit column codes $v_1 \sim v_{64}$ in main memory. Each code is chopped into separate bytes. 
The $j$-th byte of all codes are stored in a continuous memory region. 
In other words, the bytes of a code  $v_i$ are vertically distributed in different memory regions. Each of these regions is called a byte slice.
At query time, a 256-bit memory word such as $W_1$ is loaded into the CPU 
and processed with byte-level AVX2 SIMD instructions, achieving 32-way parallelism.
Similar to VBP, ByteSlice enjoys the benefit of early stop during scan.
% , and is shown
% to be more efficient when the data is uniformly distributed~\cite{byteslice}.
For example, the scan on Block 1 ($v_1 \sim v_{32}$) in Figure~\ref{fig:byteslice}
may stop early after processing $W_1$, without needing to load and process $W_2$.
%The use of byte-level SIMD instructions, as opposed to bitwise instructions, 
%reduces the number of instructions needed to perform a scan.

ByteSlice is also efficient at lookup. It distributes a $k$-bit code across 
$\lceil k/8\rceil$ memory words. In Figure~\ref{fig:byteslice}, a lookup on $v_i$
will incur at most 2 memory accesses, as opposed to 11 in VBP. The number 
$\lceil k/8\rceil$ is typically small (between 1 to 3), so it can be overlapped with
other instructions in the CPU's instruction pipeline~\cite{byteslice}.
Nonetheless, ByteSlice has not exploited any data skew found in real data. 
All codes
in Figure~\ref{fig:byteslice} are encoded into two bytes regardless.

\section{PP-VBS} \label{sec:SkewScan}

%Following part can be copied from GRF proposal.

%1) Plainly describe the prefix-free property.

%2) Show that enforcing the prefix-free property will increase the tree height and the code length variance, using the Huffman figure here. Larger tree height and code length variance hurts the memory efficiency under VBP because more zeros are needed to make up the differences. 

%Although important in communication, the prefix-free property is not crucial in column stores as long as we can differentiate codes that share prefix elsewhere. Therefore, we propose designing a prefix preserving encoding (PPE) scheme that gets rid of the prefix-free property. 

%As mentioned, ByteSlice is better than VBP in terms of time-efficiency and it uses 8-bits, i.e., a byte, as the operating unit.  Consequently, we propose adopting a PPE encoding tree with a fan-out of $2^8 = 256$. Use a figure to illustrate.

%Propose a new storage layout named VBS, give a figure to show such structure.

PE-VBP leverages data skew in real data
but suffers from poor lookup performance.
ByteSlice has excellent balance between scan and lookup
but has not leveraged data skew to accelerate its operations.
In this section, we present PP-VBS (Prefix Preserving Variable Byte Slice).
PP-VBS aims 
to leverage data skew to accelerate scans without jeopardizing the efficiency of lookups.
Similar to ByteSlice and PE-VBP, 
PP-VBS
is a suite of techniques to 
encode column values into integer codes
(Section~\ref{sec:skewscan:vbe}), store the codes' bytes in main memory using 
a specialized layout
(Section~\ref{sec:skewscan:vbs}), 
and perform efficient scan and lookup operations
on top (Section~\ref{sec:skewscan:qp}).

\subsection{Prefix Preserving Encoding (PPE)}
\label{sec:skewscan:vbe}
We begin by arguing that 
the \emph{prefix-free} property, commonly found in
variable-length encoding \cite{huffman,informationTheory,optimal_search_tree,minimal_binary_tree} and used in PE-VBP, is actually unnecessary 
in the context of main-memory data processing
but subverts the performance of both scans and lookups.
Traditionally, variable-length encoding schemes were designed for reducing the 
communication cost when sending a message across the network \cite{informationTheory}. 
In that context, codes are often concatenated as a sequential byte stream, so the
prefix-free property is crucial for ensuring the message can be decoded without
ambiguity. 
Concretely, a prefix-free encoding scheme is usually constructed by building a \emph{frequency-sorted tree} \cite{huffman}. 
Figure \ref{fig:huffmantree} shows one such example to encode 13 distinct letters
of varying frequencies.
This tree encodes the values in multiples of two bits, thus having a fan-out of four.
%In real implementation, it would be multiples of bytes so has a fan-out of 256.
The numbers on top of the letters are their frequencies. The bits in the dashed boxes or circles are partial codes. A letter's code is the concatenation of the bits from
the root to its slot. For example, `C' $= (0010)_2$, `K' $= (10)_2$.
When building the tree, values of higher frequencies are assigned to higher slots,
thus shorter codes. Scarce values are assigned to lower slots and thus longer codes.
The prefix-free property stipulates that a (shorter) code cannot be a prefix of 
another (longer) code. For example, we cannot use $(\underline{10}01)_2$ to encode a letter because $(10)_2$ is already assigned to `K'.
So, in Figure~\ref{fig:huffmantree}, a slot that
contains a letter cannot have a pointer to a sub-tree at the same time, nor vice versa. 

Code lengths have direct impact on both scan and lookup performance.
% more iterations and more instructions during scans.
Firstly, longer codes generally imply more iterations and more instructions during scan.
Second, it also implies longer lookup time.
By contrast, \emph{prefix-preserving} encoding (PPE) generally produce shorter codes on average.
% Hence, prefix-free encoding has an intrinsic trade-off between short codes and 
% maximum code length --- the more short codes it assigns, the longer the maximum code
% length (i.e., tree height) is. In Figure~\ref{fig:huffmantree}, the maximum code 
% length is 8 bits, or four levels.
Figure~\ref{fig:4way} shows a prefix-preserving example for the same data in 
Figure~\ref{fig:huffmantree}. These are still variable-length codes, e.g,. 
`A'=$(0001)_2$, `C'=$(01)_2$.
% `E'=$(0110)_2$,  `G'=$(10)_2$, `J'=$(1011)_2$. 
A slot can both contain a letter and have a pointer to a sub tree.
Note that `C'=$(01)_2$ does \emph{not} prevent other codes using it as the prefix, e.g.,
`E'=$(\underline{01}10)_2$. 
% As a result, we can encode three letters, in contrast to one
% in Figure~\ref{fig:huffmantree}, using one-byte codes. 
% This in turn reduces the
% number of values to be encoded in the lower levels, bounding the tree height
% at 2, in contrast to 3 in Figure~\ref{fig:huffmantree}.
% Table~\ref{tab:codelen} compares prefix-free and prefix-preserving on three columns
% from real data sets. We see XXXXXXXXXXXXXX.
%With prefix-preserving, we can bound all codes' length at four bits (two levels) in this example.
In general, PPE achieves smaller frequency-weighted \emph{average code length}
than prefix-free.
The average code length in Figure~\ref{fig:4way} is 2.31 bits as opposed to 3.19 bits in Figure~\ref{fig:huffmantree}. 
%In other words, 
Even though the PPE trees are not always as well balanced as this example, 
prefix-preserving still allows more values to use short codes
and fewer values to use long codes.
That in turn allows more values to leverage data skew to accelerate scans
and fewer values to suffer from inefficient lookup.
The only question is how to deal with ambiguity
for codes that share the same prefix. 
Fortunately, main-memory storage and communication are two very different problems.
We can actually store some lightweight disambiguation information elsewhere in
the storage level (Section \ref{sec:skewscan:vbs}).

% \begin{figure*}
% \begin{minipage}[b]{.5\linewidth}
% \centering
% \includegraphics[width=0.80\columnwidth]{figures/huffman}
% \caption{Prefix-free Encoding Tree for Values `A' to `M'}
% \label{fig:huffmantree}
% \end{minipage}
% \begin{minipage}[b]{.5\linewidth}
% \includegraphics[width=0.8\columnwidth]{figures/4way}
% \caption{Prefix-preserving Encoding Tree for Value `A' to `M'.}
% \label{fig:4way}
% \end{minipage}
% \end{figure*}

\begin{figure}
\centering
\includegraphics[width=0.7\columnwidth]{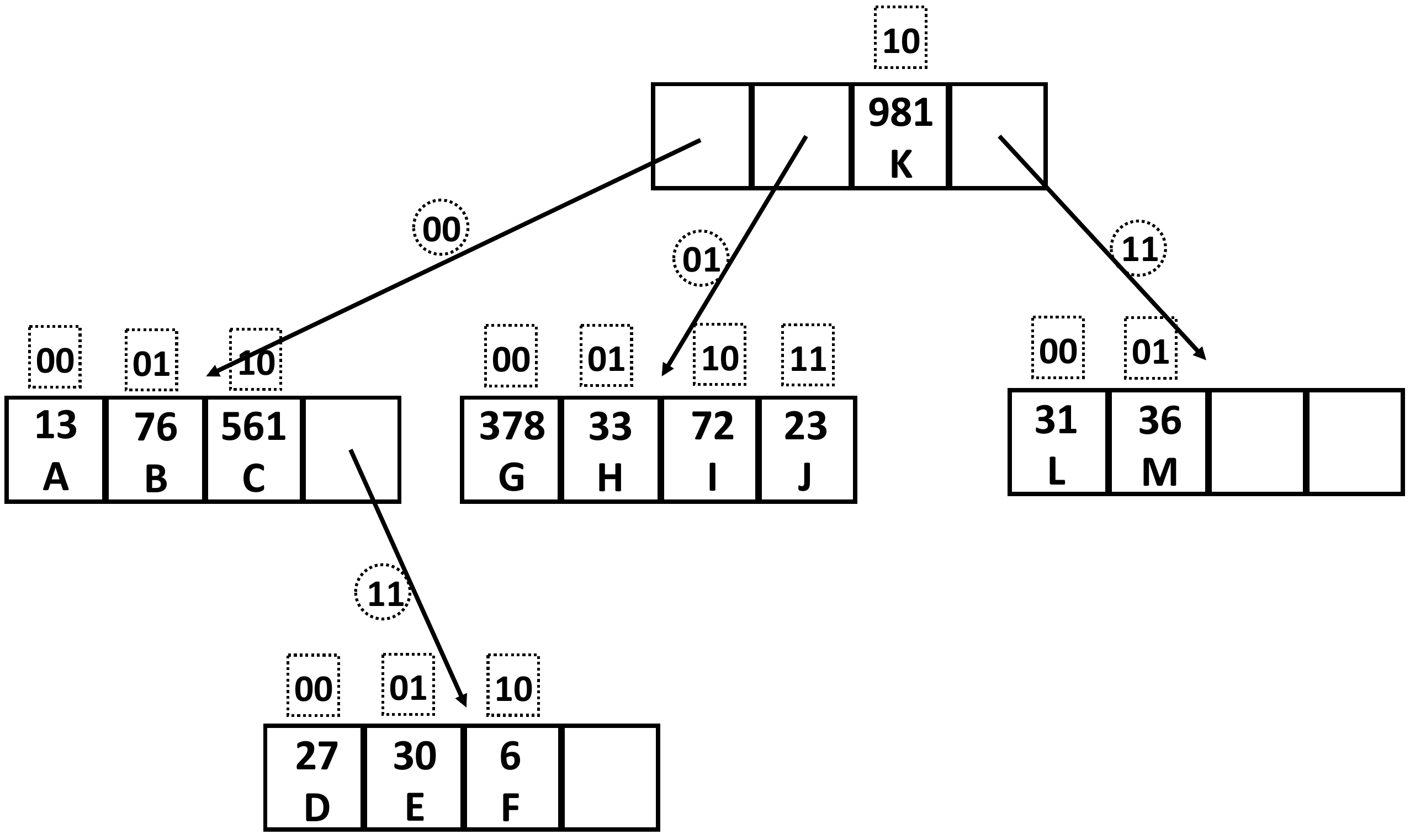}

\caption{Prefix-free Encoding Tree for Values `A' to `M'}
\label{fig:huffmantree}
\vspace{-0.4cm}
\end{figure}

\begin{figure}
\centering
\includegraphics[width=0.7\columnwidth]{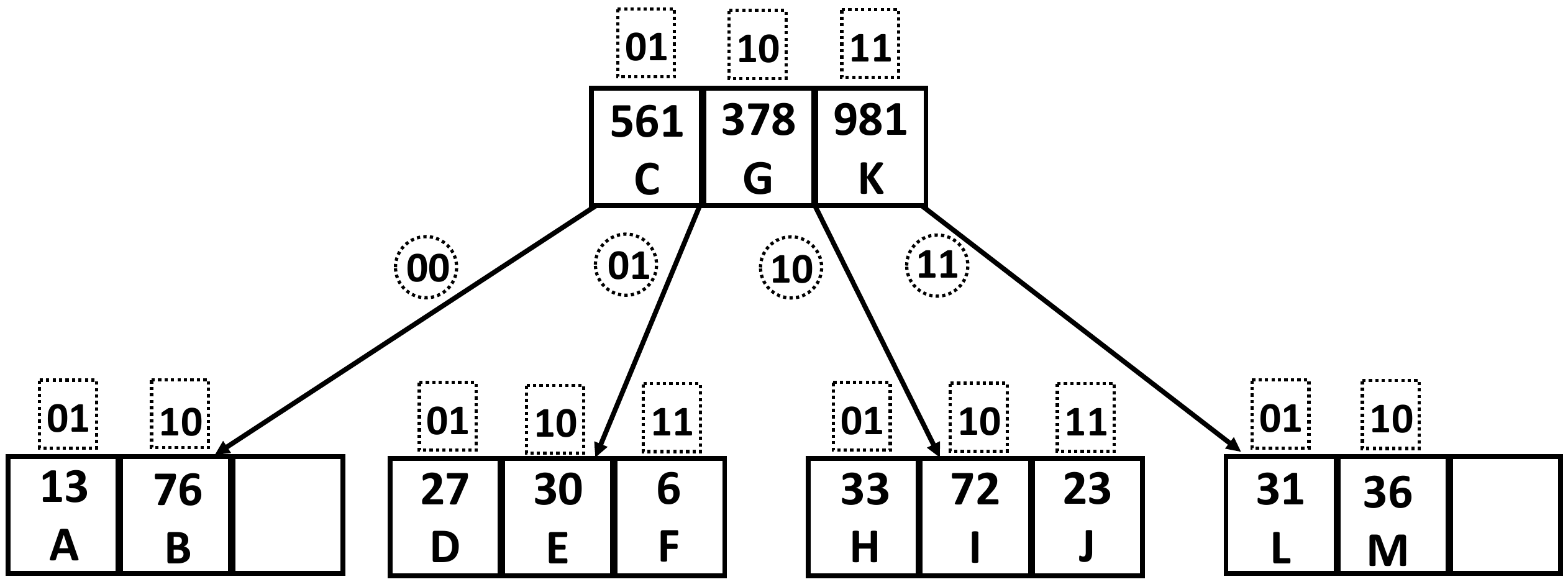}
%\vspace{-2ex}
\caption{Prefix-preserving Encoding Tree for Value `A' to `M'.}
\label{fig:4way}
\vspace{-0.4cm}
\end{figure}

PPE is based on building a prefix-preserving $M$-way encoding tree.
%Figure \ref{fig:4way} illustrates when $M=4$. 
In Figure \ref{fig:4way}, $M=4$.
Each node in the tree contains $M-1$ slots. Each slot from the left to the right represents
a sub-code from $1$ to $M-1$, e.g., from $(01)_2$ to $(11)_2$ in Figure \ref{fig:4way}.
We call it as \emph{slot sub-code}, marked in a dashed box.
% Notice that we use decimal numbers or binary numbers to denote the codes interchangeably, e.g., $(255)_{10} = (11111111)_2$.
Each non-leaf node has $M$ pointers connecting to $M$ child nodes. Each pointer from the 
left to the right carries a sub-code from $0$ to $M-1$, e.g., from $(00)_2$ to $(11)_2$ in
Figure \ref{fig:4way}. We call it \emph{pointer sub-code}, marked in a dashed circle.
Note that slot sub-codes do
not start from 0 for the sake of disambiguation, as explained later.

Given an encoding tree, we obtain the encoding of a value $v$ by concatenating all
pointer sub-codes on the path from the root to $v$, and $v$'s own slot sub-code.
In Figure \ref{fig:4way}, the code of value `C' is simply its slot 
sub-code $(01)_2$, since it is in the root node; the code of `A' is $(0001)_2$, which is formed
by concatenation of pointer sub-code $(00)_2$ and slot sub-code $(01)_2$. 
Figure~\ref{fig:4way} also illustrates the intuition of our encoding method: we encode high-frequency
values (`C', `G', `K') with shorter codes, and low-frequency values with longer codes.

\subsubsection{\small{\bf Encoding Numerical Columns (Order-preserving)}}

In addition to being prefix-preserving, the encoding tree in Figure~\ref{fig:4way}
is also \emph{order-preserving}. That is, the tree's
in-order traversal gives the ordered sequence of the encoded values.
It enables the ordinal 
comparison of two column values (e.g., `A' vs. `C') by comparing their codes directly
(e.g., $(0001)_2$ vs. $(01)_2$). When comparing two integer codes of different lengths,
we treat them \emph{as if} the shorter code were padded zeros at the end to align with the 
longer one. Thus, $(0001)_2 < (01\underline{00})_2 \Leftrightarrow$   `A' $<$ `C'.
For this reason, we do not use 0 as slot sub-codes. Otherwise, it cannot differentiate
between a short code with padding 
$(\underset{slot}{\underline{01}}\phantom{0}\underset{pad}{\underline{00}})_2$ and a long code
$(\underset{pointer}{\underline{01}}\phantom{0}\underset{slot}{\underline{00}})_2$. 
Notice this is different from prefix-free, as there are still prefix-sharing codes
(e.g., `C' and `E'). Order-preserving is crucial for numerical columns as it enables
efficient evaluation of all types of predicates ($=,\neq, <, >, \leq, \geq$) 
without decoding codes into their represented values.

\begin{algorithm}[h]
\footnotesize
\caption{PPE for Numerical Columns}
\label{algo:vbeNumerical}
	\KwIn{ $n$ distinct values $A$ = $\{A_0 < A_1 < ... < A_{n-1}\}$, value weights $Q = \{Q_0, Q_1, ..., Q_{n-1}\}$. 
         }
	\KwOut{$n$ value codes $C$ = $\{C_0, C_1, ..., C_{n-1}\}$, code lengths $L$ = $\{L_0, L_1, ..., L_{n-1}\}$.}
	
	\SetKwFunction{VBE}{PPE\_NUMERICAL} 
	\SetKwProg{Fn}{Function}{:}{} 
    
    \tcp{\scriptsize{$b$: the level of the parent node}}
    \tcp{\scriptsize{$s$: the start index}}
    \tcp{\scriptsize{$e$: the end index}}
	\Fn{\VBE{$C$, $L$, $b$, $s$, $e$}}
	{
	   \tcp{\scriptsize{two terminating conditions}} 
    	\If{$e - s < 256$ \textbf{OR} $b \ge 2$}  
    	{
    	    $\beta = \lceil \log_{256}{(e - s + 1)} \rceil$; \\
    	    \For{$i = s$, ..., $e-1$}
    	    {
    	        $C_i = (C_i \ll \beta * 8) + 1 + (i - s)$; \newline
    	        $L_i = b + \beta$; 
    	    }
    	}
   	\Else{
        	Find the largest 255 frequencies in $Q_s\sim Q_{e-1}$, store their indices in an ascending array $I$. \\
        	\For{$t = 0$, ..., $254$}{
       	    $L_{I_t} = b + 1$; \\
        	    $C_{I_t} = (C_{I_t} \ll 8) + t + 1$ \\
        	}
        	\For{$t = 0$, ..., $253$}
        	{
        	    \For{$i = I_{t}+1$, ..., $I_{t+1}-1$}
        	   {
        	        $C_i = (C_i \ll 8) + t + 1$; \\
        	    }
        	    \Call{\VBE}{$C$, $L$, $b+1$, $I_t+1$, $I_{t+1}$}; \\
        	}
        	\tcp{\scriptsize{Deal with the leftmost sub-tree}}
        	\For{$i$ = $s$, ..., $I_0$ - $1$}
        	{
        	    $C_i$ = $C_i$ $\ll$ $8$;
        	}
        	\Call{\VBE}{$C$, $L$, $b+1$, $s$, $I_0$}; \\
         	\tcp{\scriptsize{Deal with the rightmost sub-tree}} 
        	\For{$i$ = $I_{254}+1$, ..., $e-1$}
        	{
        	    $C_i$ = ($C_i$ $\ll$ 8) + 255; \\
        	}
        	\Call{\VBE}{$C$, $L$, $b+1$, $I_{254}+1$, $e$}; \\
   	}
    }
    \textbf{for}($i = 0$, ..., $n-1$): $C_i = 0$; \\
	\textbf{for}($i = 0$, ..., $n-1$): $L_i = 0$; \\
	\Call{\VBE}{$C$, $L$, $0$, $0$, $n$}; \\

\end{algorithm}

To ensure that we do not underutilize the available short codes, all slots of each non-leaf node in a prefix-preserving encoding tree should be filled with higher-frequency values.
For example, the root node in Figure \ref{fig:4way} is full of values. 
Therefore, we build the encoding tree in a depth-first search manner and fill all non-leaf nodes in the paths from the root node to leaf nodes.

Without loss of generality, given a sequence of distinct values in ascending order 
$A_0 < A_1 < \cdots < A_{n-1}$, and their corresponding frequencies 
$Q = \{Q_0, Q_1, \cdots, Q_{n-1}\}$, %assuming $n > 256$, 
Algorithm \ref{algo:vbeNumerical} shows the pseudocode of constructing a prefix-preserving encoding (PPE) for numerical columns.
The outputs of Algorithm \ref{algo:vbeNumerical} are the codes $C$ and the corresponding code lengths $L$. 
Algorithm \ref{algo:vbeNumerical} builds a 256-way tree, here $M$ equals to 256(=$2^8$), since the code lengths are multiple of 1 byte (8 bits) in our setting.
%The algorithm builds a 256-way encoding tree. % in a depth-first search manner.
At first, we initialize the elements of $C$ and $L$ as $0$'s (Lines 22-23).
Then Algorithm \ref{algo:vbeNumerical} invokes a function \texttt{PPE\_NUMERICAL} 
recursively to construct the encoding tree. 
Among the parameters passed to \texttt{PPE\_NUMERICAL}, $s$ and $e$ represent the \emph{start} and \emph{end} index, which indicate that values $A_s \sim A_{e-1}$ will be encoded in 
a sub-tree by this call; $b$ indicates the level of its parent node.
At the first round, \texttt{PPE\_NUMERICAL} selects the 255 largest frequencies from $Q_0 \sim Q_{n-1}$ and stores their corresponding array indices in an ascending array $I$ (Line 8), i.e., $I_0 < I_1 < ... < I_{254}$. 
%For example, assume that $Q_m$ is the largest frequency in $Q$, then $I_0 = m$.  
These most frequent 255 values will reside in the root node and be encoded as 1-byte codes (Lines 9-11), i.e., $(1)_{10} \sim (255)_{10}$. 
This ensures the most frequent values are encoded with the shortest 1-byte codes.
% It is worth noting that we fill all the slots of the root node in order to let as many frequent values as possible to be assigned 1-byte (shortest) codes. 
% We believe that scan and lookup performance can benefit from this strategy.
After that, the values between adjacent slots in the root node will be encoded by
building a sub-tree, e.g., the values in the positions [$I_1+1, I_2$) of the sequence $A$, by calling 
\texttt{PPE\_NUMERICAL} recursively (Lines 12-15). 
Note the extra efforts to deal with the leftmost and rightmost sub-trees (Lines 16-21).
Since a non-leaf node must be full before we create any child
nodes by recursion, we can see:
\begin{property}
All (non-leaf) internal nodes in a \emph{prefix-preserving encoding} tree are full.
\label{p:full}
\end{property}
%This ensures we do not underutilize available short codes.

The recursion has two terminating conditions (Line 2):

%\begin{itemize}
\noindent
[{\bf $e - s < 256$}]: It means a leaf node has enough slots to hold all the values $A_s \sim A_{e-1}$ pending to be encoded.
%those values in the positions [$s,e$) of $A$  in a leaf node totally. 
Then their slot sub-code will be assigned $(1)_{10} \sim (e-s)_{10}$ (Line 5), and their code length will be $b+1$ bytes (Line 6), where $b$ is the level of its parent node.
There is no need to compare the frequencies of these values.
 
\noindent
[{\bf $b \ge 2$}]: 
%Note that the lengths of codes will affect the lookup and scan performance under ByteSlice and VBS as well. 
Under some extreme column distributions, the encoding tree will become unbalanced, which has a negative impact on scan and lookup performance.
In order to reduce the maximum code length and alleviate the tree unbalance, after 2 levels, %the tree is built regularly, 
the values in the positions [$s,e$) are packed into one leaf node even though the number of values is larger than $255$ (Lines 3-6). 
Since in real-world datasets, the number of distinct values of columns
(aka domain size) is mainly in the range of ($2^7, 2^{16}$)~\cite{li2013bitweaving,willhalm2009simd}, we believe that the value $2$ is a reasonable threshold.
%\end{itemize}
%After obtaining the codes $C$ and code lengths $L$, as said, we pad zeros at the end of short codes to align with the longest codes according $L$ to obtain the fixed-length codes $C'$.

%We believe this strategy won't affect the memory-efficiency much, since in real-world datasets, the number of distinct values of columns is mainly in the range of ($2^7, 2^{16}$) \cite{willhalm2009simd}.
%\end{itemize}

\subsubsection{{\small{\mbox{\bf Encoding Categorical Columns (Non Order-preserving)}}}}

While all previous works achieve order-preserving for all columns, 
we observe extra opportunity by forgoing this property for categorical columns.
Typically, categorical columns (e.g, city name, brand name, color) are only involved 
in equality predicates ($=, \neq$), but not range predicates ($<, >, \le, \ge$). 
By omitting the need to handle range predicates, we can build an encoding tree
that is more balanced and remove exceptionally long codes.
%Discuss:
%1) Time complexity of this encoding method, do a comparison to the time complexity of PE. 
%This is also optimal in memory cost.
% In real datasets, a categorical column often have some frequent values that makes up 
% a large portion of the column. 
% A simple example is the brands of cars, several brands (e.g., Ford, Toyota) occupy huge market share. 
% While encoding categorical columns, order-preserving is not necessary.

Given a sequence of distinct values $A = \{A_0, A_1, ..., A_{n-1}\}$ with descending 
frequencies $Q = \{Q_0 > Q_1 > ... > Q_{n-1}\}$, %assuming $n > 256$, 
Algorithm \ref{algo:vbeCategorical} shows the pseudocode of PPE for categorical columns. 
Similarly, the outputs of Algorithm \ref{algo:vbeCategorical} are the codes $C$ and corresponding code lengths $L$. 
The algorithm builds the encoding tree in a breadth-first search manner.
Intuitively, it starts by assigning the 255 most frequent values to the top-level
slots (i.e., root node). Then, it assigns the next $256\times 255$ 
(pointers$\times$slots) most frequent values to the second level. 
If there are still values to  encode, it assigns them to the third level, 
for up to $256^2\times 255$ (pointers$^2\times$slots), and so forth.
Therefore, Algorithm~\ref{algo:vbeCategorical} bounds the height of the tree
by a balanced tree that can hold all values (Line 1).
It first fills in the non-leaf nodes level-by-level
(Lines 3--8). When filling the possibly non-full last level, it evenly distributes 
the remaining values among the last level nodes to maximize entropy (Lines 9--15).
The order does not matter between codes, assuming no range comparison 
is needed in the workload. 
Property~\ref{p:full} holds true for categorical columns. In addition:

\begin{property}
A PPE tree for categorical columns is balanced.
\end{property}

% \subsubsection{\textbf{Codec}}
% After obtaining the codes $C$ and code lengths $L$, we use a hash map to facilitate 
% encoding of incoming column values. For simplicity, we use a standard hash map using
% standard integers (\texttt{int64}) as keys.
% So we pad zeros at the end of $C$ to make them fixed-length $C'$.
% We then build a hash map $A_i \mapsto C'_i$. Recall that slot sub-codes cannot be zero.
% So we know all trailing zero bytes of $C'_i$ are padding and we strip them before
% further processing. Decoding uses a hash map similarly.

\begin{algorithm}[]

%\small
\footnotesize
\caption{PPE for Categorical Columns}
\label{algo:vbeCategorical}
	\KwIn{  
	$n$ distinct values $A = \{A_0, A_1,..., A_{n-1}\}$,  value weights  $Q = \{Q_0 > Q_1 > ... > Q_{n-1}\}$.
	}
	\KwOut{
	$n$ value codes $C = \{C_0, C_1,..., C_{n-1}\}$, code lengths $L = \{L_1, L_2,...,L_{n-1}\}$.
	}
	
    \tcp{\scriptsize{B: the number of the levels of the balanced tree}}
	$B = \lceil \log_{256}(n+1) \rceil$; \\
	$\alpha = 0$; \\
	\For{$b$ = $1,...,B-1$}
	{
		%$shift$ := ($max\_byte$ - $byte$) * 8\\
		\For{$c = 0,...,256^{b-1}-1$}
		{
			\For{$i=1,...,255$}
			{
				$C_{\alpha}$ = $c \ll 8$ + $i$; \\
				$L_{\alpha}$ = $b$; \\
				$\alpha$ += $1$;\\
			}
		}
	}
	\For{$i=1,...,255$}
	{
		\For{$j=0,...,256^{B-1} - 1$}
		{
			$c_{\alpha}$ = $j \ll 8 + i$;\\
			$L_{\alpha}$ = $B$; \\
			$\alpha$ += $1$;\\
			\If{$\alpha == n$}
			{
				exit(0);\\
			}
		}
	}
\end{algorithm}

{\bf String columns}: It is worth noting that certain string columns will encounter range predicates (e.g.,$<$,$>$) in analytical queries. 
We regard such string columns as semi-categorical and treat them as same as the numerical columns. 
In other words, we construct an order-preserving PPE for a semi-categorical column using Algorithm \ref{algo:vbeNumerical}. 
%As mentioned, for predicate involving similarity search (i.e., {\tt LIKE}) on string columns, codes are decoded before a scan is evaluated in the traditional way.
%In addition, it is intractable when a string column encounters similarity search (i.e., {\tt LIKE} predicate), 
%since it is impossible to carry out a {\tt LIKE} predicate on the codes directly.
In addition, since it is impossible to carry out a similarity search ({\tt LIKE} predicate) on the codes directly, in our implementation, the codes are decoded to original strings before a similarity search is evaluated in the traditional way.

\subsection{Variable Byte Slice (VBS)}
\label{sec:skewscan:vbs}

\begin{figure}[]
\centering 
\includegraphics[width=1.0\columnwidth]{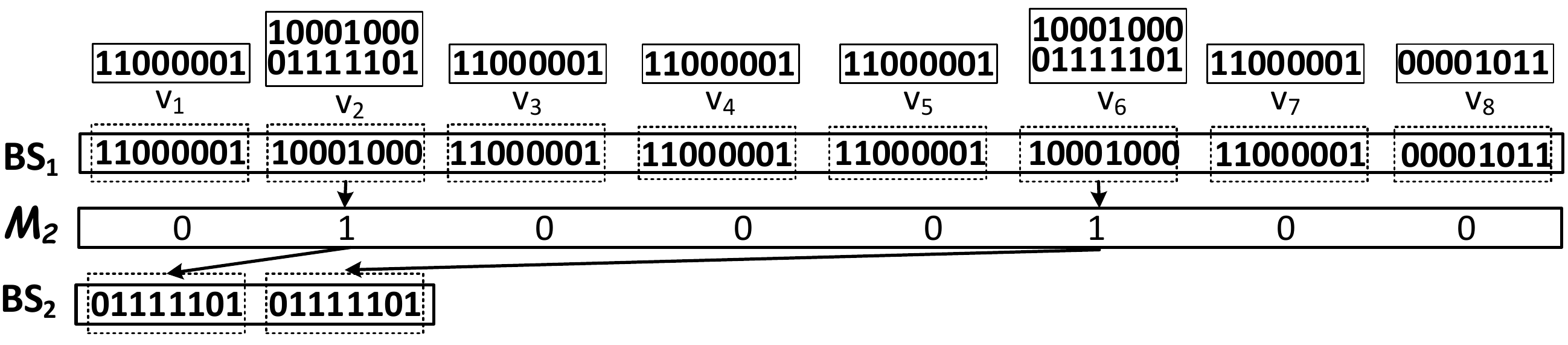}
\caption{Variable Byte Slice (VBS)}
\label{fig:vbs}
\vspace{-0.4cm}
\end{figure}

In above, PPE encodes column values into variable-length integer codes
\emph{logically}. Next, Variable Byte Slice (VBS) describes how to store their bits 
in the main memory \emph{physically} so that we can perform efficient scan and 
lookup operations.

Figure~\ref{fig:vbs} gives a quick glimpse of the VBS storage format used in PP-VBS.
VBS arranges the $j$-th byte of all codes together
in a separate memory region $BS_j$. 
%However, padding bytes are \emph{not} stored.
To disambiguate, a bit mask $M_j$ indicates whether a code has a materialized
$j$-th byte in $BS_j$. VBS inherits all good properties of ByteSlice ---
being amenable to byte-level SIMD instructions and efficient lookup --- while
bringing in skew-awareness to further improve performance.
VBS introduces irregularity that was previously challenging for SIMD processing. 
We will show how PP-VBS utilizes advanced bit manipulation instructions 
(e.g., Intel BMI2) in query processing to address those challenges.
%(Section~\ref{sec:skewscan:qp}).
%{\bf above}

Suppose a column of values are encoded with prefix preserving encoding (PPE), VBS vertically distributes 
the bytes of codes into several contiguous memory regions.
Same as~\cite{byteslice}, we refer to these memory regions as byte slices.
If the maximum code length is $K$ bytes, then $K$ byte slices are needed:
$BS_1, BS_2,$..., $BS_K$. 
Starting from the second byte slice, each byte slice is accompanied with a
bitmask indicating whether a code has a (materialized) byte in this slice:
$\mathcal{M}_2,..., \mathcal{M}_K$. 
In other words, VBS packs the $j$-th ($j \ge 2$) bytes of the codes which have the $j$-th byte into $BS_j$ tightly. 
Since all codes have at least one byte, $\mathcal{M}_1$ is omitted.
Figure~\ref{fig:vbs} illustrates the first $8$ codes of a column under VBS. The
maximum code length is two bytes, so there are $2$ byte slices and $1$ bitmask. 
We use the notation $BS_j^{(i)}$ to denote the $i$-th byte in $BS_j$. 
For example, in Figure \ref{fig:vbs}, $BS_1^{(2)}=(10001000)_2$, $BS_2^{(1)}=(01111101)_2$. Similarly, we use notation $\mathcal{M}_{j}^{(i)}$ to denote the $i$-th bit in $\mathcal{M}_j$:
$\mathcal{M}_2^{(1)}={(0)}_2$, $\mathcal{M}_2^{(2)}=(1)_2$. 
In this example, $v_2$ and $v_6$ have the second byte. Their second bytes are packed together in byte slice $BS_2$ with $\mathcal{M}_2^{(2)}$ and $\mathcal{M}_2^{(6)}$ indicating that information.

\begin{figure}[htbp]
\centering 
\includegraphics[width=.8\columnwidth]{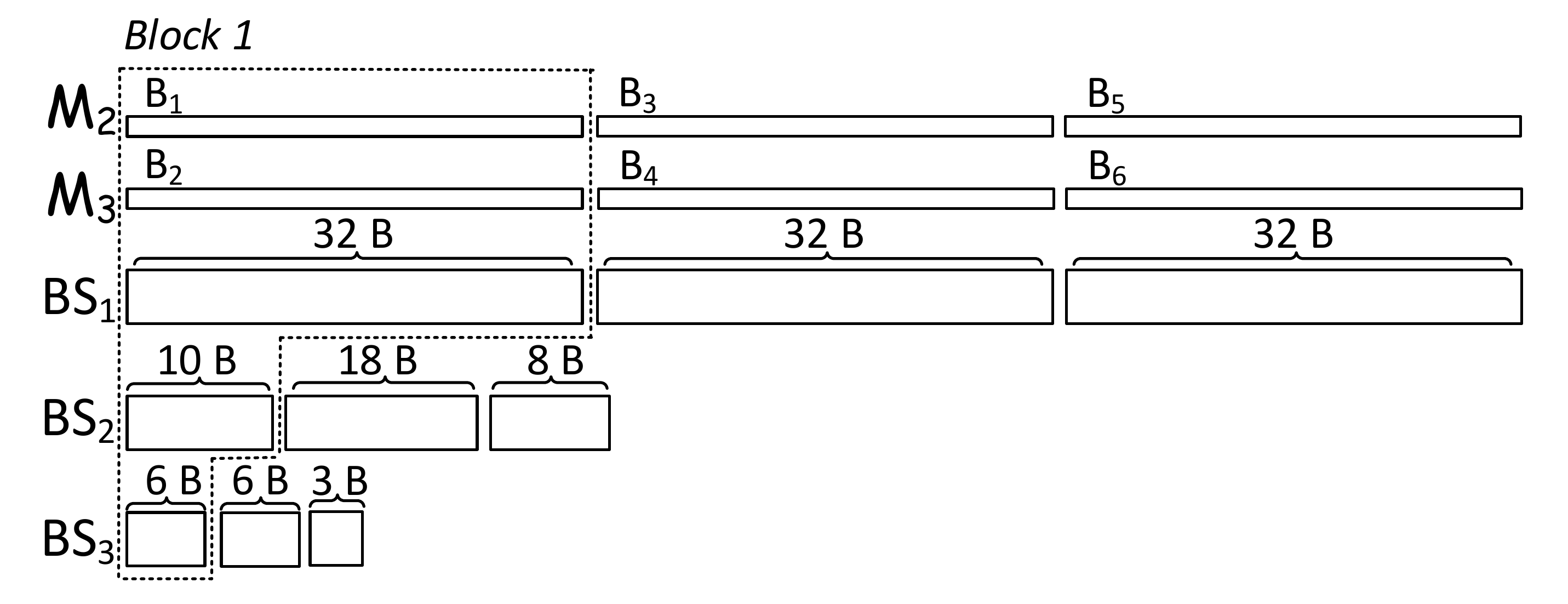}
\caption{A VBS Formatted Column}
\label{fig:vbs2}
\vspace{-0.2cm}
\end{figure}

VBS further divides a column into many \emph{blocks}. Each block contains data
belonging to 32 consecutive codes, under the current implementation using
256-bit AVX2 instructions.
Figure \ref{fig:vbs2} illustrates a VBS column having three blocks (96 codes),
with maximum code length of 3 bytes. Each horizontal line in the figure
represents a contiguous memory region.
The dashed line outlines all data belonging to Block 1. Not surprisingly, 
Block 1 takes up 32 bytes in $BS_1$. Among the 32 codes in Block 1, only 10 
have the second byte. So Block 1 takes up only 10 bytes in $BS_2$. Ditto for
$BS_3$. In general, $BS_j$ always has fewer bytes than $BS_{\{<j\}}$.
When data is moderately and highly skewed, a large portion of codes are
one byte, and $BS_{\{\ge 2\}}$ have very few bytes. 
A segment of 32 bits (4 bytes) from the bitmasks $\mathcal{M}_2$ and 
$\mathcal{M}_3$ are used to store this sparsity information of Block 1. 
In summary,  Block 1 in VBS uses $4\times 2 + 32 + 10 + 6 = 56$ bytes, 
whereas it would have used  $32 \times 3 = 96$ bytes in the original
ByteSlice. 

% Since every code has the first byte, each block occupies 32 bytes (i.e., a 256-bit word) in $BS_1$. 
% For example, Block 1 takes up $BS_1^{(1)} \sim BS_1^{(32)}$.
% Among the 32 codes of Block 1, only 10 have the second bytes, so block 1 takes up 10 bytes in $BS_2$, i.e., $BS_2^{(1)} \sim BS_2^{(10)}$. 
% A 32-bit word $B_1$ in the bitmask $\mathcal{M}_2$ indicates which codes have the second bytes, we call such {$S/8$}-bit word as \emph{indicator word}. 
% In addition, we can use {\tt POPCNT} instruction to count the number of $1$'s in $B_1$ ({\tt POPCNT}($B_1$)), i.e., {\tt POPCNT}($B_1$) = $10$.
% The next block (i.e., Block 2) has 18 codes that have the second bytes, so it occupies the next contiguous 18 bytes in $BS_2$, i.e., $BS_2^{(11)} \sim BS_2^{(28)}$. 
% Similarly, Block 3 takes up $BS_2^{(29)}\sim BS_2^{(36)}$. $B_3$ and $B_5$ serve as indicator words for Block 2 and 3 respectively. 
% In this example, Block 1, 2 and 3 have 6, 6 and 3 codes that have the third bytes respectively and $B_2$, $B_4$, $B_6$ serve as the indicator words. 
% It is worth noting that for each block, the number of third bytes would never be larger than the number of second bytes, since if a code has the third byte, it must has the second byte.

As discussed in \cite{byteslice, li2013bitweaving, imprint, columnSketches}, scan is a memory-bound operation. Thus, VBS offers at least two advantages:
%\begin{itemize}
	%\item 
	1) It reduces memory bandwidth consumption between the memory and the CPU
	during scan, making it scale on many-core architectures more efficiently.
	2) The CPU cache can contain data of more codes, which increases the cache hit rate.% under the help of hardware prefetching.
%\end{itemize}

%Specifically, an $S$-bit memory word contains bytes from $S/8$ different codes. 

%The magic number $8$ comes from the smallest bank width of SIMD multi-operand instructions. 
%A bank width of $8$ bits implies the highest degree of SIMD parallelism, i.e., $S/8$-way (e.g., $256/8=32$), is exploited.
%The choice also has the advantage of simple implementation because bytes are directly addressable. 

\subsection{Scan and Lookup}
\label{sec:skewscan:qp}
%scan . .
%look up..
%complex predicat is simple: follow byteslce, pipelne bitmp to tnext operator [cite byteslice]

A PP-VBS scan takes in the VBS column, the predicate operation, and the literal code. 
It outputs a result bit vector, which indicates the matching rows.
%It can output a bit vector indicating matching rows or a list of matching row IDs.
%In this paper, the scan outputs a bit vector by default. 
%In general, for very low selectivities a row ID list should be used and for higher selectivities a bit vector should be used for the scan output. 
%This is because at high selectivities the row ID list format requires large amount of memory movement while retrieving the codes from columns of interest, since a row ID is represented by a 32-bit or 64-bit record number in most systems \cite{byteslice, li2013bitweaving, li2014widetable}.
%From our experiments, a query with a very low selectivity is always scan-dominant, which means the performance loss of a lookup with a bit vector instead of a row ID list would not affect the performance of the whole query.
Our scan algorithm fully utilizes the data parallelism provided by SIMD instructions.
Here, we assume AVX2 architecture is used and the SIMD register is 256-bit.
A SIMD instruction carries out the same operation on the vector of banks simultaneously. 
For example, the \texttt{\_mm256\_add\_epi8()} instruction performs a 32-way addition between two SIMD registers, which adds 32 pairs of 8-bit integers simultaneously.
We further exploit advanced bit manipulation instructions such as \texttt{pdep()} and
\texttt{pext()} in new generation CPUs to process the bit masks.

% In addition, modern CPUs are equipped with various bit manipulation instructions. 
% For example, \texttt{pdep()} and \texttt{pext()} instructions are new generalized bit-level compress and expand instructions. 
% We use them to manipulate the bitmasks of \emph{Variable ByteSlice} in the scan and lookup algorithms.

\subsubsection{\textbf{Scan}}
%The scan algorithm, pseudocode.
%+ BMI2

% We now describe how to evaluate predicates ($<, >, \le, \ge, =, \ne$, {\tt BETWEEN}) 
% using SIMD instructions under PP-VBS. As mentioned, the output of such filter scans is
% a \emph{result bitvector} $R$, which indicates which codes in the column satisfy 
% the predicate. 
% Each scan operation processes $256/8=32$ codes at a time. 

Without loss of generality, we explain with the \emph{GREATER-THAN}
predicate in the form of  $v > c$. Other predicate types follow suit. 
We use the notation $v^{[j]}$ to denote the $j$-th byte of code $v$. For example,
in Figure \ref{fig:vbs}, $v_1^{[1]}=(11000001)_2$, $v_2^{[2]}=(01111101)_2$. 
Suppose we have four column codes and a predicate literal $c$ as below:
\centerline{$\overline{v}_1=(10000001~\phantom{00000000})_2$}
\centerline{$\overline{v}_2=(11000000~\phantom{00000000})_2$}
\centerline{$\overline{v}_3=(01110111~10101001)_2$}
\centerline{$\overline{v}_4=(10000001~10001001)_2$}
\centerline{ $c = (10000001~\phantom{00000000})_2$}
The goal of the scan is to determine whether each $\overline{v}$ passes or fails 
the predicate $\overline{v} > c$.
Recall that when comparing two codes of different length, we treat the short
code \emph{as if} it were padded zeros at the end. 
PP-VBS starts by comparing the first byte of $c$ with the first byte of
all $\overline{v}_i$:

\centerline{$\overline{v}_1^{[1]} = c^{[1]}$ and $\overline{v}_2^{[1]} > c^{[1]}$ and $\overline{v}_3^{[1]} < c^{[1]}$ and $\overline{v}_4^{[1]} = c^{[1]}$}

At this point, we can safely conclude that $\overline{v}_2$ passes the predicate,
while $\overline{v}_3$ fails. There is no need to examine 
$\overline{v}_2$'s and $\overline{v}_3$'s
second byte.
For $\overline{v}_1$, we proceed to check its auxiliary bit in $\mathcal{M}_2$
(not shown here) to learn that it has no second byte. Hence we know for a fact
that $\overline{v}_1 = c$, meaning $\overline{v}_1$ fails the predicate.
For $\overline{v}_4$, we check its auxiliary bit in $\mathcal{M}_2$, and learn that
it has a second byte. \emph{Without} fetching and examining that second byte
$\overline{v}_4^{[2]}$, we immediately know for a fact that $\overline{v}_4 > c$,
meaning it passes the predicate. That is because, as described in 
Section~\ref{sec:skewscan:vbe}, the last byte (slot sub-code) of a code cannot
be zero. Therefore, 
$$\overline{v}_4 = 10000001~\underset{>0}{\underline{********}} > c = 10000001~\underset{pad}{\underline{00000000}}$$

In summary, scans on all above codes \emph{stop early} after the first byte,
even though the maximum code length is two. More formally, scan
cost on PP-VBS can be bounded by the following lemma:

\begin{lemma}\label{lemmaonly}
Let $l(v)$ be the byte-length of code $v$.
For all predicate types $\texttt{op} \in \{<,>,\le,\ge,=,\neq\}$,
the evaluation of predicate $v_1~\texttt{op}~v_2$ conclusively stops after examining
the $m$-th byte, where $m \le \min\big( l(v_1), l(v_2)\big)$.
\end{lemma}

\begin{proof}
Let $P_l(v)$ denote the $l$-byte prefix of a code $v$. 
By definition, $v = P_{l(v)}(v)$ for all $v$.
% we have $v_1 = P_{l(v_1)}(v_1)$ and $v_2 = P_{l(v_2)}(v_2)$. 
Without loss of generality, we assume $l(v_1) < l(v_2)$. 
If we have $P_{l(v_1)}(v_1) \le P_{l(v_1)}(v_2)$, then $v_1$ must be less than $v_2$, because $v_2$ must have at least one trailing byte that is larger than 0. 
If $P_{l(v_1)}(v_1) > P_{l(v_1)}(v_2)$, we have $v_2 < P_{l(v_1)}(v_2)+1 \le P_{l(v_1)}(v_1) \le v_1$. 
Another special situation is when $l(v_1) = l(v_2)$. Under such case, when we have $P_{l(v_1)}(v_1) < P_{l(v_1)}(v_2)$, then $v_1 < v_2$; when $P_{l(v_1)}(v_1) = P_{l(v_1)}(v_2)$, then $v_1 = v_2$; when $P_{l(v_1)}(v_1) > P_{l(v_1)}(v_2)$, then $v_1 > v_2$. 
Thus, we can compare the prefixes of length $\min(l(v_1), l(v_2))$ leading bytes to obtain the comparison result between $v_1$ and $v_2$. This completes the proof.
\end{proof}

{\bf In addition}, PP-VBS inherits the early stop capability of ByteSlice,
even when both $v$ and $c$ are long. To illustrate, consider a different 
predicate literal $c'$ for the above example:

\centerline{$c' = (10000000$ $01010000)_2$}

\noindent
In this case, after comparing the first byte, we obtain:

\centerline{$\overline{v}_1^{[1]} > c'^{[1]}$ and $\overline{v}_2^{[1]} > c'^{[1]}$ and $\overline{v}_3^{[1]} < c'^{[1]}$ and $\overline{v}_4^{[1]} > c'^{[1]}$}

\noindent
which suffices to conclude that $\overline{v}_1$, $\overline{v}_2$ and
$\overline{v}_4$ pass the predicate, whereas $\overline{v_3}$ fails. There is no
need to compare $c'^{[2]}$ with any code.

\begin{algorithm}[]
\footnotesize
\caption{Vectorized Scan on PP-VBS}
\label{algo:scan}
	\KwIn{predicate $v>c$; $c$ is $l$-byte;}
	\KwOut{result bit vector $R$}
	\For {j = 1,...,l}
	{
		$W_{cj} = \texttt{simd\_broadcast}(c^{[j]})$\\
	}
	$I_j = 0$, for $1 \le j \le l$ \\
	\For{every $S/8$ codes $v_{(I_1+1)}...v_{(I_1+S/8)}$}
	{
		$\mathcal{Z}_{ep} = 1^{S/8}$ \\
		$\mathcal{Z}_{gt} = 0^{S/8}$ \\
		\For{$j=2,...,\min(l+1,K)$}
		{
		    $B_j = \texttt{load}(\mathcal{M}_j^{(I_1+1)},...,\mathcal{M}_{j}^{(I_1+S/8)})$
		}

		\For{$j=1,...,l$}
		{
			\If{$\texttt{test\_zero}(\mathcal{Z}_{eq})$} 
			{
				\textbf{break} 
			}

			$W_j$ = $\texttt{simd\_load}(BS_j^{(I_j+1)},...,BS_j^{(I_j + S/8)})$ \\
			$M_{gt}$ = $\texttt{simd\_movemask\_epi8}(\texttt{simd\_cmpgt\_epi8}(W_j, W_{cj}))$ \\
			$M_{eq}$ = $\texttt{simd\_movemask\_epi8}(\texttt{simd\_cmpeq\_epi8}(W_j, W_{cj}))$ \\
			\If{$j > 1$}
			{
				$M_{gt} = \texttt{pdep}(M_{gt}, B_j)$ \\
				$M_{eq} = \texttt{pdep}(M_{eq}, B_j)$ \\
			}
			\If{$j < l$}
			{
				$\mathcal{Z}_{gt} = \mathcal{Z}_{gt} | (\mathcal{Z}_{eq} \& M_{gt})$ \\
				$\mathcal{Z}_{eq} = \mathcal{Z}_{eq} \& M_{eq} \& B_{j+1}$
			}
			\ElseIf{$l < K$}
			{
				$\mathcal{Z}_{gt} = \mathcal{Z}_{gt} | (\mathcal{Z}_{eq} \& M_{gt}) | (\mathcal{Z}_{eq} \& M_{eq} \& B_{j+1})$
			}
			\Else 
			{
				$\mathcal{Z}_{gt} = \mathcal{Z}_{gt} | (\mathcal{Z}_{eq} \& M_{gt})$
			}
		}
		$I_1 = I_1 + S/8$ \\
		\For {$j=2,...,l$}
		{
	    	$I_j = I_j + \texttt{POPCNT}(B_j)$ 
		}
		Append $\mathcal{Z}_{gt}$ to $R$\\
	}

\end{algorithm}

Algorithm \ref{algo:scan} delineates the VBS's implementation that
``vectorizes'' the above scan algorithm to evaluate $S/8=32$ codes in parallel
in each iteration using SIMD AVX2 instructions, where the register length $S=256$ bits.
The algorithm takes predicate $v > c$ as input and outputs a result bit vector $R$ whose $i$-th bit is $1$ if $v_i$ satisfies the predicate or $0$ otherwise. 
Suppose the constant code $c$ is $l$-byte ($1 \le l \le K$), $K$ is the maximum code length and the column is numerical.
Initially, the bytes of the constant code $c$ are broadcast to $l$ SIMD words (Lines 1-2). 
%Since different codes have different lengths, and $BS_j$ packs tightly the $j$-th bytes of codes which has the $j$-th byte, it is inevitable that different bytes of one code have different indices in corresponding byte slice.
%In Figure \ref{fig:vbs}, the index of $v_6^{[1]}$ in $BS_1$ is $6$, while the index of $v_6^{[2]}$ in $BS_2$ is $2$.  
%Therefore, we have to use $I_j$ to  the index of $BS_j$ and $I_j$ ($1 \le j \le l$) is initialized as $0$ (Line 3). 
Then, the algorithm scans the column codes in a {\tt for} loop, with one block (32 codes) each time (Lines 4-24). 
% For each block, the algorithm first prepares
% two $\frac{S}{8}$-bit result masks $\mathcal{Z}_{eq}$ and $\mathcal{Z}_{gt}$ (32-bit masks). 
The $i$-th bit of the 32-bit mask $\mathcal{Z}_{gt}$ holds the temporary
Boolean value of $v_i > c$. Ditto for $\mathcal{Z}_{eq}.$ These two masks are
iteratively updated as we examine more bytes of the codes.
% The mask $\mathcal{Z}_{eq}$ is similarly interpreted for the condition $v=c$.
In order to know which codes in the current block have the $j$-th byte, 
it loads corresponding auxiliary word $B_j$ from $\mathcal{M}_{j}$ (Lines 7-8).
%Here, $2 \le j \le l$, because there is no $BM^{K+1}$ to indicate whether the codes have the $K+1$-th byte. 
The algorithm then examines the codes byte-by-byte through $l$ iterations (Lines 9-24). 
In the $j$-th iteration, it first tests whether $\mathcal{Z}_{eq}=0$.
If yes, the predicate is conclusive for all codes in the block, and thus
we can \emph{stop early} (Lines 10-11). 
%If not, it loads the $j$-th bytes of $S/8$ codes from index $I_j$ in $BS_j$. 
Otherwise, it loads the next $S/8$ bytes from $BS_j$ into $W_j$,
and executes two $S/8$-way SIMD instructions to determine the codes whose $j$-th bytes are either (1)$>$$c^{[j]}$ or (2)$=$$c^{[j]}$. 
The 256-bit results are condensed into 32-bit bit maps using the
\texttt{movemask} instruction (Lines 13--14).
% After SIMD comparisons, results are condensed to two $S/8$-bit masks, which are put into two local masks $M_{gt}$ and $M_{eq}$. 
% This is done by converting each bank of all $1$'s($0$'s) into a single bit of $1$($0$) using SIMD {\tt movemask} instruction (Lines 13-14). 

Note that $W_j$ may contain bytes belonging to the next block due to sparsity.
Concretely, only $T_j$ = {\tt POPCNT}($B_j$) $\le 32$ codes have $j$-th 
($j>1$) byte in current block.  Thus, only the $T_j$ lower bits of $M_{eq}, M_{gt}$ in 
Lines 13--14 are of interest.
To put these bits in the right place, we utilize the instruction {\tt pdep}\footnote{{\tt pdep}(src, mask), which is an Intel BMI2 instruction, takes the low bits from the first {\tt src} operand and deposits them into a destination at the corresponding bit locations that are set in the second {\tt mask} operand. All other bits (bits not set in {\tt mask}) in destination are set to zero.}  to scatter the $T_j$ consecutive low order bits of $M_{gt}$ and $M_{eq}$ according to $B_j$ (Lines 15-17). 
Figure \ref{fig:pdep-pext}(a) shows an example of {\tt pdep} instruction.
After that, if the $i$-th bit of $B_j$ is zero, then the $i$-th bit of $M_{gt}$ and $M_{eq}$ will be $0$.
Otherwise, the $i$-th bit in $M_{gt}$ or $M_{eq}$ will indicate whether $j$-th byte of $i$-th code in a block is (1)$>$$c^{[j]}$ or (2)$=$$c^{[j]}$.
After processing the current $S/8$ codes, the algorithm must update the indices $I_j$ ($1 < j \le l$) to indicate where to load the $j$-th bytes of next block from $BS_j$ (Lines 25-27). Finally, the result of current block $\mathcal{Z}_{gt}$ is appended to $R$ (Line 28) before the processing of the next block begins.
Algorithm \ref{algo:scan} can be easily modified to handle other comparison operators (e.g., $\ge$, $=$), which are omitted in interest of space.
%The details are given in Appendix \ref{++++ appendix}. 

\begin{figure}[]
\centering 
\includegraphics[width=1.0\columnwidth]{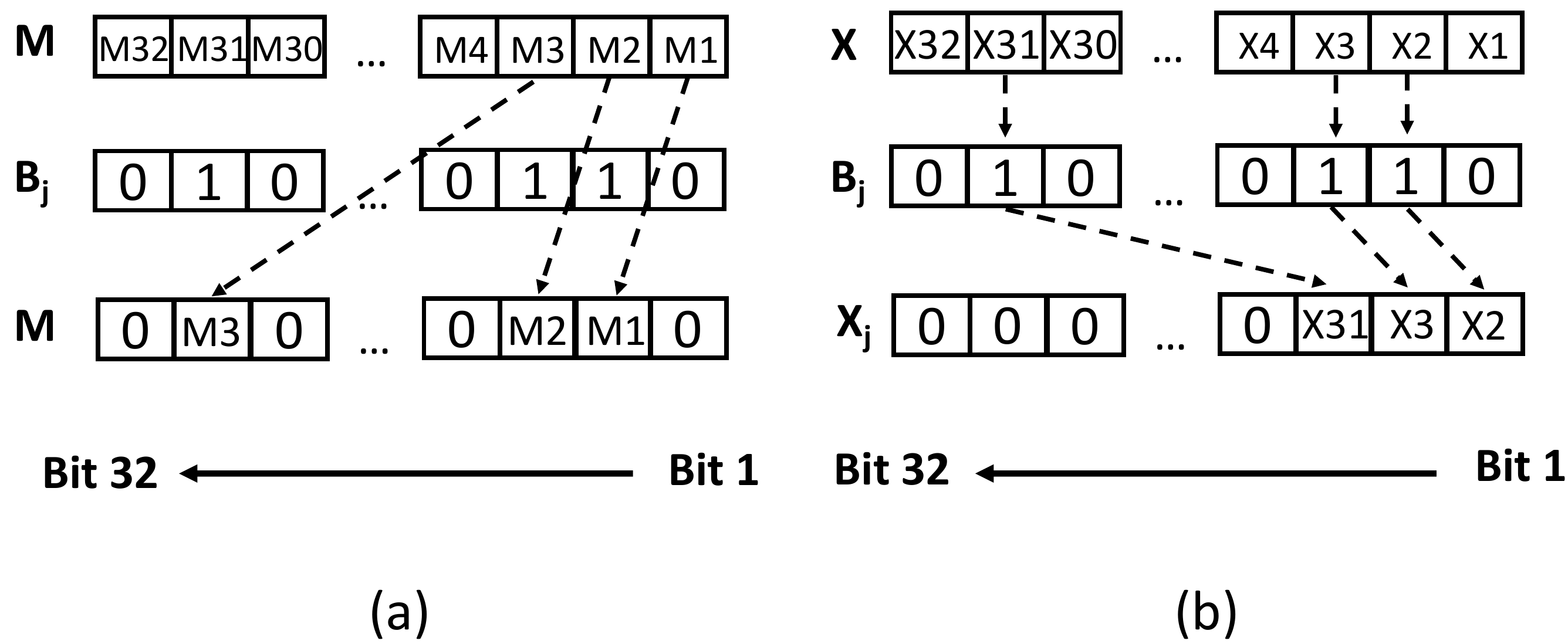}
\caption{(a) {\tt pdep} instruction (b) {\tt pext} instruction}
\label{fig:pdep-pext}
\vspace{-0.4cm}
\end{figure}

Lemma~\ref{lemmaonly} offers insight why PP-VBS is efficient for skewed columns.
Early stop in Algorithm~\ref{algo:scan} happens \emph{per block}. Scan on a block
can stop early when all $S/8$ codes in it stop early. When the column is more skewed,
the probability is higher that all codes in one block are high-frequency values,
which are in turn encoded with few bytes under PPE. That bounds the scan cost
on every block at a small $m$, even if the literal code $c$ is long.

% Even though the length of the constant predicate is large (e.g., $l = K$), scans on skewed columns also benefit from the design of SkewScan.
% Suppose $l^*$ as the maximum code length of the codes stored in a SkewScan block, then the predicate evaluation on this block can be terminated at iteration $\min(l^*, l)$ according to lemma \ref{lemmaonly}. 
% Since most values in a skewed column are encoded as short codes, it is highly possible that $l^* < l$. Therefore, the scan on this block can stop at $l^*$ iterations. 

To handle complex predicates that involve multiple columns, we follow 
ByteSlice~\cite{byteslice} to pipeline the result bit vector of one predicate evaluation 
to another so as to increase the early stop probability of the subsequent evaluation. 

% Another issue is how complex predicates that involve multiple columns are evaluated under VBS. 
% Following ByteSlice \cite{byteslice}, we pipeline the result bit vector of one predicate evaluation to another so as to increase the early stopping probability of the subsequent evaluation. 
% In our context, same as ByteSlice, there are two possible implementations of such pipelining approach, \emph{column-first} and \emph{predicate-first}.
% The discussion of these two implementations is similar to that in ByteSlice \cite{byteslice}, so we omit it here.

\subsubsection{\textbf{Lookup}}

%The lookup algorithm, pseudocode.
%Given a result bit vector, then retrieve the codes from the corresponding columns.
\emph{Lookup} refers to the operation of retrieving the codes 
from a column of interest given a result bit vector $R$, which is produced by the scan.
%Very often, $R$ is 
%produced by scan and is sparse.
In the scan-based OLAP framework mentioned in Section \ref{sec:background}, the retrieved codes are inserted into an array of a fixed-length 
data type (e.g., \texttt{int32[]}). 
Subsequent operations like aggregations and sorts will consume this array. 
In order to easily embed PP-VBS into this framework, we pad zeros at the end of 
short codes in the {lookup} operation to generate fixed-length codes.

\begin{algorithm}[ht]
\footnotesize
\caption{Lookup Under PP-VBS}
\label{algo:lookup}
	\KwIn{$R$: result bit vector; $\mathbb{C}$: a VBS formatted column; $K$: the maximum code length in $\mathbb{C}$}
	\KwOut{$L$: a list of $K$-byte codes}
	\For{$1 \le j \le K$}{
	    $I_j = 0$;
	}
	\For{{every $\frac{S}{8}$-bit word $x$ \textbf{\textup{in}} $R$}}
	{
	    \For{$j = 2,...,K$}{
	        $B_j = \texttt{load}(\mathcal{M}_{j}^{I_1+1},...,\mathcal{M}_{j}^{I_1+S/8})$;
	    }
		$\beta = 1$; \\
		\For{$j=2,...,K$}
		{
			\If{$\texttt{test\_zero}(x\& B_j)$}
			{
				\textbf{break};
			}
			$\beta = \beta + 1$;\\
		}
		\For{$j=2,...,\beta$}{
			$x_j = \texttt{pext}(x,B_j)$;
		}
%		\If{$\beta == 1$}
%		{
%			\While{$!test\_zero(x)$}
%			{	
%				$r := I_1 + POPCNT(\textbf{P}(x)) + 1$;\\
%				$C := BS_1^{(r)}$;\\
%				$x = \textbf{R}(x)$;\\
%				Append $C << (K - 1) * 8$ to $L$; \\
%			}
%		} 
%		\Else
%		{
		\While{$!test\_zero(x)$}
		{
			
			$i = \texttt{POPCNT}(\textbf{P}(x)) + 1$;\\
			$r_1 = I_1 + i$; \\
			$C = BS_1^{(r_1)}$;\\
			$x = \textbf{E}(x)$; \\
			\For{$j=2,...,\beta$}
			{	
				\If{$\texttt{test\_zero}(B_j$ $\&$ $(1 \ll S/8 - i))$}
				{
					$C = C \ll (\beta - j + 1) * 8$; \\
					\textbf{break};
				}	
				$r_j = I_j + \texttt{POPCNT}(\textbf{P}(x_j)) + 1$;\\
				$C = (C \ll 8) + BS_j^{(r_j)}$; \\		
				$x_j = \textbf{E}(x_j);$
			}
			$C = C \ll (K - \beta) *8$;  \\
			Append $C$ to $L$;\\
		}
%		}
		$I_1 = I_1 + S/8$;\\
	    \For{$j=2,...,K$}{
	        $I_j = I_j + \texttt{POPCNT}(B_j)$;
	    }
	}

\end{algorithm}

Algorithm \ref{algo:lookup} delineates the pseudo-code of the lookup operation 
under PP-VBS. It can be seen as the ``inverse operation'' of VBS.
Before formally describing Algorithm 4,  we first introduce two bitwise manipulations that provide a fast way to find and manipulate the rightmost 1 bit in a word:
% The notations $\textbf{E}(x)$, and $\textbf{P}(x)$ are used to denote the two 
% manipulations respectively. Among them,
(1) $\textbf{E}(x)=x\&(x-1)$, which \textbf{e}rases the rightmost $1$ in a word $x$;
(2) $\textbf{P}(x)=x\oplus -x$, $\oplus$ is XOR, which \textbf{p}ropagates the
rightmost 1 to the left in the word $x$, making them all 1's, and erases the rightmost $1$. 
Here are examples on an 8-bit word $x$.
\centerline{$x=(010100\underline{1}0)_2$}
\centerline{$\textbf{E}(x)=x\& (x-1) = (010100\underline{0}0)_2$}
\centerline{$\textbf{P}(x)=x\oplus -x = (111111\underline{0}0)_2$}

The basic idea behind Algorithm \ref{algo:lookup} is to extract each $1$ from the 
result bit vector $R$ and reconstruct the corresponding code from the 
VBS-formatted column $\mathbb{C}$.
To construct a code, we retrieve the bytes from corresponding byte slices and then concatenate them in order.
For example, in Figure \ref{fig:vbs}, to obtain 2-byte $v_2$, bytes $BS_1^{(2)}$ and $BS_2^{(1)}$ will be retrieved. 
When the column is more skewed, the probability is higher that the length of a code retrieved is short so that fewer bytes will be retrieved from the memory.

% First, we initialize $I_j$ ($1 \le j \le K$) as 0 (Lines 1--2). 
The algorithm executes in a {\tt for} loop (Lines 3--26) that handles the
result bit vector $R$ in blocks of 32 bits. 
Let $K$ be the maximum code length in $\mathbb{C}$.
$I_j$ ($1 \le j \le K$) tracks the start position of the current block
in $BS_j$.
In each iteration, it loads a $\frac{S}{8}=32$-bit word $x$ from $R$ into CPU.
%since the result bit vector $R$ is an array of $\frac{S}{8}$-bit words and block has $S/8$ codes. 
%In order to know whether the $S/8$ codes associated with $x$ have $j$-th($j>1$) bytes, the algorithm has to load the next $\frac{S}{8}$ bits from $\mathcal{M}_j$ as $B_j$(Lines 4-5). 
In order to know which codes in the current block have the $j$-th ($j>1$) byte, the algorithm loads the associated auxiliary word $B_j$ from $\mathcal{M}_j$ (Lines 4--5) .
It then determines the (local) maximum length of codes which are extracted 
in the current block (Lines 7--10) as $\beta$. 
So it knows which byte slices will be touched and 
skips processing the rest. 
Next, the algorithm uses {\tt pext}\footnote{{\tt pext}(src, mask) is an Intel BMI2 instruction. 
For each bit set in the second {\tt mask} operand, it extracts the corresponding bits from the first {\tt src} operand and writes them into contiguous low bits of destination. } 
instruction to pack bits from the word $x$ according to $B_j$ ($2 \le j \le \beta$) 
into contiguous low bits of a destination word $x_j$ (Lines 11--12).
Figure~\ref{fig:pdep-pext}(b) shows an example of the {\tt pext} instruction.
% Note that the remaining upper bits of $x_j$ are zeroed. 
A bit set in $x_j$ indicates one byte will be retrieved from $BS_j$ for constructing a code that is selected in $R$.

%\begin{figure}[]
%\centering 
%\includegraphics[width=1.0\columnwidth]{figures/pext-pdep.pdf}
%\caption{{\tt pext} instruction}
%\label{fig:pext}
%\end{figure}

Then we iteratively find the rightmost `1' in $x$, calculate its position using
\texttt{POPCNT}$(\textbf{P}(\cdot))$, and erase it from $x$
(Lines 14--17).
Each popped `1' indicates a code to be reconstructed.
% In each iteration of the {\tt while} loop over word $x$, we first retrieve the first byte of a code (Lines 14--17). 
% To do so, the algorithm propagates rightmost $1$ to the left in the word x (i.e., \textbf{P}(x)), and then use the {\tt POPCNT} instruction to count the number of 1's in it (i.e., {\tt POPCNT}(\textbf{P}(x))). 
% It calculates `{\tt POPCNT}(\textbf{P}(x)) + 1' to find the offset of the rightmost
% (lowest) 1 in x (Line 14). 
The offset is added to tracking index $I_1$ to obtain $r_1$, which is the position of that code's first byte
in $BS_1$ (Line 15).
% After loading first byte into code $C$ (Line 16), it removes the rightmost $1$ from $x$ (Line 17). 
% Then, it retrieves the remaining bytes (if any) of the code with a {\tt for} loop 
% (Lines 18--24). 
For the $j_{\{>1\}}$-th byte, it first inspects whether the code has the $j$-th byte.
% If not, the algorithm can \textbf{break} (Line 21). 
If so, it retrieves the $j$-th byte from $BS_j$ (Lines 22--24) and concatenates it to $C$. 
After constructing $C$, we pad zeros at its end and append $C$ to the list
$L$ (Lines 25--26). 
The above procedure repeats until $x=0$, meaning all wanted
codes are reconstructed.
After one $\frac{S}{8}$-bit word $x$ has been consumed, we update every 
tracking indices $I_j$ ($1\le j \le K$).

%When the column is more skewed, the probability is higher that all codes in one block are high-frequency values, which are encoded as short codes under PPE. 
%Then $\beta$ is smaller than $K$, this is one reason that lookup under PP-VBS is even more efficient than ByteSlice when the column is skewed.
%Since in current block, lookup under PP-VBS only touches $\beta$ byte slices, while ByteSlice has to retrieve bytes from $K$ byte slices.

 \section{Column-Layout Advisor} \label{section:cla}
\begin{figure*}
    \begin{center}
        \subfigure[Column ``Seconds'' (Numeric)]{\label{fig.choose_skewscan}
            \includegraphics[width=0.60\columnwidth]{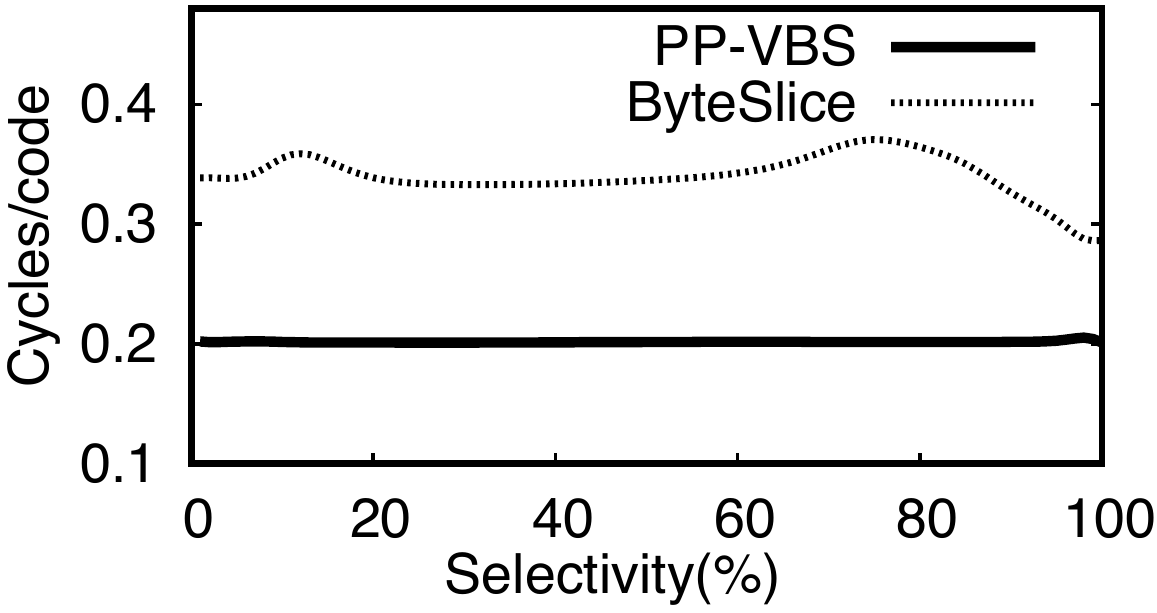}}
        \subfigure[Column ``Total'' (Numeric)]{\label{fig.choose_anyone}
           \includegraphics[width=0.60\columnwidth]{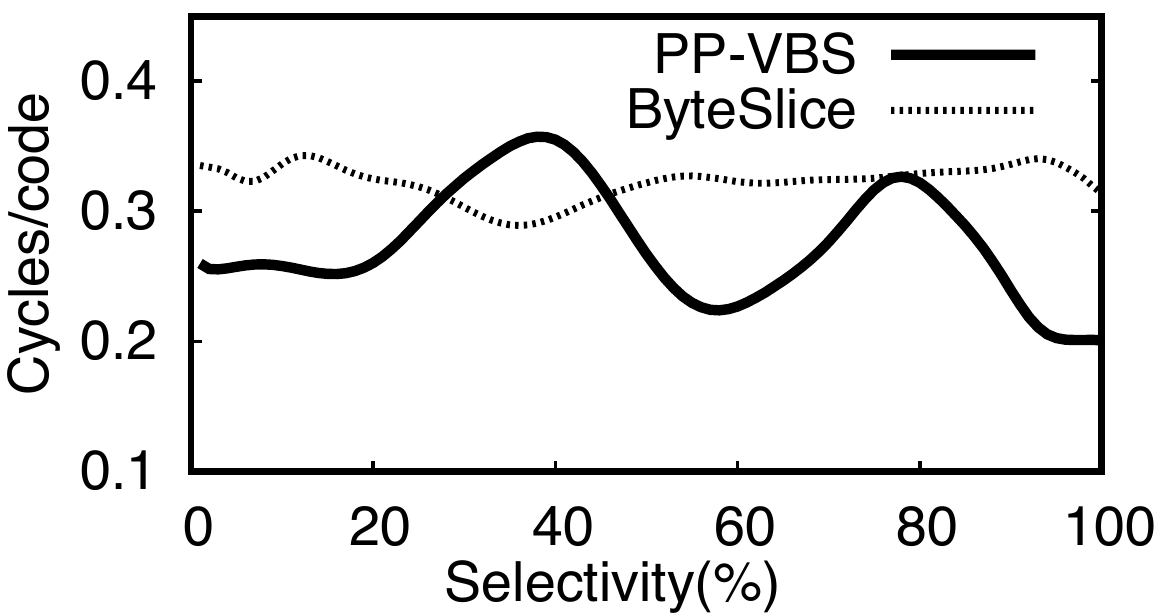}}
        \subfigure[Column ``Pickup Census Tract'' (Categorical)]{\label{fig.choose_byteSlice}
            \includegraphics[width=0.60\columnwidth]{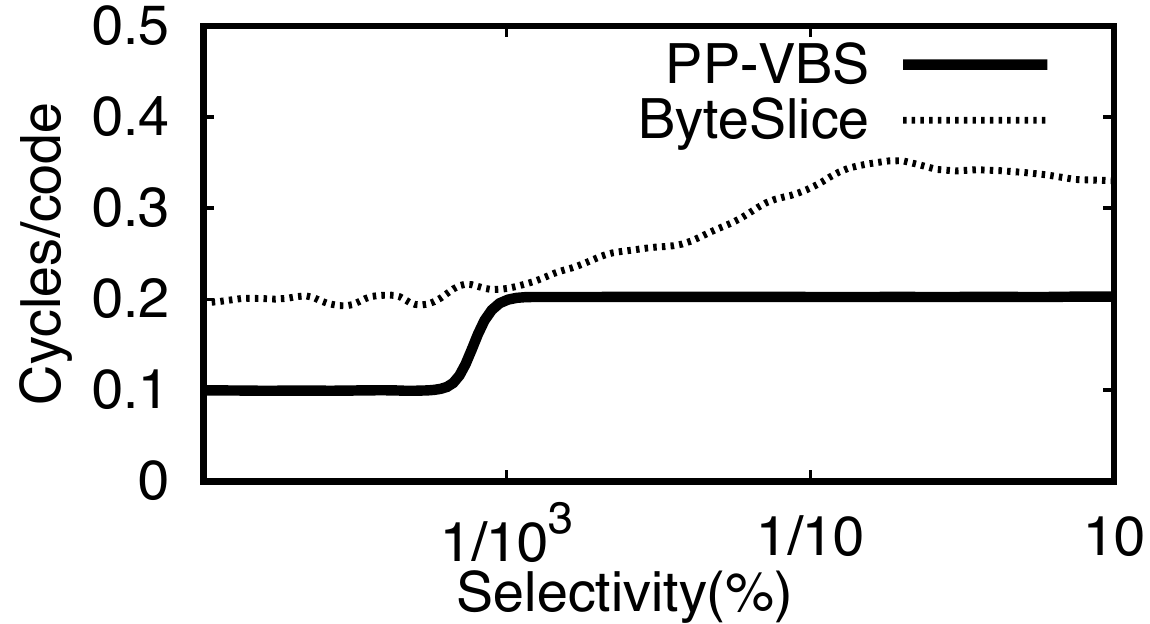}}
    \caption{Scan Profiles for 3 Real Data Columns from the TAXI\_TRIP \cite{taxitrip}}\label{fig:hybrid}
\end{center}

\end{figure*}

As our experiments show in Section~\ref{sec:experiments:microbenchmark}, 
Bit-Packed and PE-VBP storage layouts are dominated by ByteSlice or PP-VBS in terms of scan and lookup performance.
Therefore, given a data column, 
ByteStore only has to make a decision to store it in ByteSlice or PP-VBS format.
Since both ByteSlice and PP-VBS leverage byte as the storage unit, 
their lookup performances are equally good.
Therefore, the decision mainly lies on which one yields better performance on scan.

There are a variety of factors that influence the scan performances of ByteSlice and PP-VBS.
Data-related factors include the data distribution type (e.g., Zipf, Gaussian), the parameters (e.g., the Zipf factor),
the domain size, and the value type (e.g., numeric or categorical).
Query-related factors 
include the predicate type (e.g., $>$, $=$)
and its selectivity (percentage of codes that satisfy the predicate).

Although modeling the relationship 
between the scan performance and the factors above  is challenging,
ByteStore does not need to do so because the data columns are actually given.
Therefore, %instead of going for an analytic or machine learning solution to build a cost model (e.g., \cite{genericCostModel, BasuLCVYSB15}),
the column-layout advisor of ByteStore goes for an \emph{experiment-driven} approach \cite{experimentDriven}
that chooses the best storage layout for each column based on running (scan) experiments on them.
Specifically, given a data column, we 
first encode it twice: one using ByteSlice and one using PP-VBS; and then 
we generate and execute scan queries with different selectivities on top
to get a profile.
Finally, we choose the storage layout of a column based on its profile.

Figure \ref{fig:hybrid} depicts the profiles of three real columns 
from the TAXI\_TRIP \cite{taxitrip} real dataset 
obtained from Google BigQuery.
%and HEALTH \cite{health} obtained from Google BigQuery.
For numeric columns, 
we use $<$ {\tt c} as the profiling predicate
because from 
Section \ref{sec:skewscan:qp}
we know that the scan implementations of other operators (e.g,. $=$) are largely similar and thus their scan performance is also similar (our experiments also confirm this). 
For categorical columns, we use $=$ {\tt c} as the profiling predicate.
Each profile is obtained by generating and executing queries 
with 100 predicate literals from the column that span across the entire feasible selectivity spectrum.
For example, Figure \ref{fig:hybrid}c shows 
that the most unselective value in the categorical column ``Pickup Census Tract'' would retrieve 9.9\% of the column.  All other values have lower selectivity than that.
% I need to replot the figure of figure 10c

% Unless the query predicate never changes and the data characteristics of a column never changes even with new data arrival,
% otherwise the selectivity of a query predicate is hardly predictable in practice.

After profiling, our column-layout advisor computes and picks the one with a smaller area under curve (AUC).
For example, it is obvious that for the column ``Seconds'' in the TAXI\_TRIP dataset (Figure~\ref{fig:hybrid}(a)),
the AUC of PP-VBS is smaller than ByteSlice (because that column is skewed).
Therefore, the column layout advisor would retain that column encoded using PP-VBS and discard the one encoded using ByteSlice.
For the column ``Total'' in the TAXI\_TRIP dataset (Figure~\ref{fig:hybrid}(b)),
our column-layout advisor would also retain the one encoded using PP-VBS because it outperforms ByteSlice for a wide range of selectivities.  
Of course, if the (range of the) selectivity of a predicate is known, it is straightforward for our advisor to include that factor into account.  
%Alternatively, if space is not an issue, we can actually store that column using both ByteSlice and PP-VBS and let the run-time to choose which one to scan on when the selectivity is estimated at run-time.
Algorithm~\ref{algo:cla} summarizes our discussion above in the form of 
pseudo-code.

% To plot the performance curve of one column, we scatter 101 points, each represents the scan performance under a specific selectivity (0\% $\sim$ 100\%), and fit them with a curve smoothly.
% Note that it is impossible to have all the 101 constant codes to control the selectivity from 0\% to 100\% for equality predicates (i.e., $=$, $\ne$),
% %For example, if a constant code $c'$ makes 100\% codes satisfy the predicate $=c'$, we cannot find any constant codes to control other selectivity. 
% follow \cite{byteslice}, we use the same set of constant codes of less than ($<$) predicate. 
% This strategy is also suit for categorical columns, since even though the ($<$) predicate cannot be applied to the original values (categorical values), it can be applied to the integer codes. 

% Under experiment-driven advisor, real data columns can be categorized into three types. 
% The three columns {\tt Seconds} (Figure \ref{fig:hybrid}a), {\tt Total} (Figure \ref{fig:hybrid}b) and {\tt Start Timestamp} (Figure \ref{fig:hybrid}c), which come , each belongs to one type respectively.
% \begin{itemize}
% \item Type A: for such columns (e.g., {\tt Seconds}), PP-VBS dominates ByteSlice under all selectivities. 
% \item Type B: for such columns (e.g., {\tt Total}), neither PP-VBS nor ByteSlice dominates his competitor. 
% \item Type C: for such columns (e.g., {\tt Start Timestamp}), ByteSlice dominates PP-VBS under all selectivities.
% \end{itemize}

\begin{algorithm}[h]
\caption{Column Layout Advisor}
\label{algo:cla}
\footnotesize
\KwIn{
$\mathcal{C}$: a $d$-column dataset, $\mathcal{C}=\{\mathcal{C}_1, \mathcal{C}_2, ..., \mathcal{C}_d\}$. 
}
\KwOut{
$\mathcal{L}$: layouts for each column, $\mathcal{L}=\{\mathcal{L}_1, \mathcal{L}_2, ..., \mathcal{L}_d\}$.
}
\For {j = 1, ..., d} 
{
Encode $\mathcal{C}_j$ with dictionary encoding and inject the column codes into ByteSlice, get $\mathfrak{C}_j$. \newline
Encode $\mathcal{C}_j$ with \emph{prefix preserving encoding} and inject the column codes into VBS, get $\mathbb{C}_j$. 

\If{$\mathcal{C}_j$ is numerical}
{
Set predicate $\mathcal{P}$ as $<$.
}
\Else
{
Set predicate $\mathcal{P}$ as $=$.
}
Select 100 values $\{v_i\}_{i=1}^{100}$ from the column $\mathcal{C}_j$ that span across the entire feasible selectivity spectrum and the corresponding selectivities are $\{s_i\}_{i=1}^{100}$. \newline
Evaluate predicate $\mathcal{P}$ on $\mathfrak{C}_j$ with literals $\{v_i\}_{i=1}^{100}$ and get 100 (selectivity, time) points, $\{(s_i, y_i)\}_{i=1}^{100}$. \newline
Evaluate predicate $\mathcal{P}$ on $\mathbb{C}_j$ with literals $\{v_i\}_{i=1}^{100}$ and get 100 (selectivity, time) points,  $\{(s_i, z_i)\}_{i=1}^{100}$.\newline
Fit the performance curve with $\{(s_i, y_i)\}_{i=1}^{100}$ and obtain the AUC (area under the curve) $\mathcal{Y}$. \newline
Fit the performance curve with $\{s_i, z_i)\}_{i=1}^{100}$ and obtain the AUC (area under the curve) $\mathcal{Z}$. 
\If{$\mathcal{Y}$ $\le$ $\mathcal{Z}$}
{
$\mathcal{L}_j$ = {\tt ByteSlice}.
}
\Else
{
$\mathcal{L}_j$ = {\tt PP} - {\tt VBS}.
}
}
\end{algorithm}

%{\bf Eric doesn't like the term type C B A here, pleaes rewrite the following.}

% For numerical columns, we evaluate the less than predicates to test the performance of ByteSlice and PP-VBS (Lines 4-5).
% For categorical columns, we evaluate the equality predicates (Lines 6-7), as mentioned, categorical columns never see range queries.
%If the current column $\mathcal{C}_j$ belongs to type C, then the advisor chooses ByteSlice (Lines 10-11).
%If $\mathcal{C}_j$ belongs to type A, then the advisor chooses PP-VBS. (Lines 12-13).
%We need extra effort to deal with type B (Lines 14-20) , since for such columns, neither layout dominates. 
%Specifically, we compare the areas under the curves of ByteSlice and PP-VBS and choose the layout whose area is smaller. 
% After executing the 100 queries respectively on ByteSlice and PP-VBS (Lines 8-10), we fit the profile curves for both of them with points $\{s_i, y_i\}_{i=1}^{100}$ and points $\{s_i, y_i\}_{i=1}^{100}$ (Lines 11-12). 
% Finally, we compare the areas under the curves of ByteSlice and PP-VBS and choose the layout whose AUC is smaller (Lines 13-16).

\section{Experimental Evaluation}
\label{sec:experiments}
 %The experiment platform configuration

We run our experiments on a rack server with a 2.1GHz 8-core Intel Xeon CPU E5-2620 v4, and 64GB DDR4 memory.
Each core has 32KB L1i cache, 32KB L1d cache and 256KB L2 unified cache. All cores share a 20MB L3 cache. 
The CPU is based on Broadwell microarchitecture and supports AVX2 instruction set.
In the micro-benchmark evaluation, we compare PP-VBS with Bit-Packed, PE-VBP and ByteSlice to show that PP-VBS is dominating on skewed columns. 
%Table \ref{table:efficiency} summarizes their properties. 
%We implement all methods in C++. All implementations are optimized using standard techniques such as prefetching.
In the real data evaluation, we show that queries on ByteStore outperform any homogeneous storage engines.
%The programs are compiled using g++ 7.3.0 with optimization flag -O3. 
Unless stated otherwise, all experiments are run using \emph{one core}.

\subsection{Micro-Benchmark Evaluation}
\label{sec:experiments:microbenchmark}
For the first experiment, we create a column with one billion numeric values. 
The column values are integer numbers in the range of [$0, 2^d$).
% , where d is a parameter to adjust the domain size, we call it as \emph{domain factor}. 
Following PE-VBP \cite{paddedEncoding}, we generate the column values from the Zipf distribution.

\begin{figure*}[t]
\center
\includegraphics[width=0.9\linewidth]{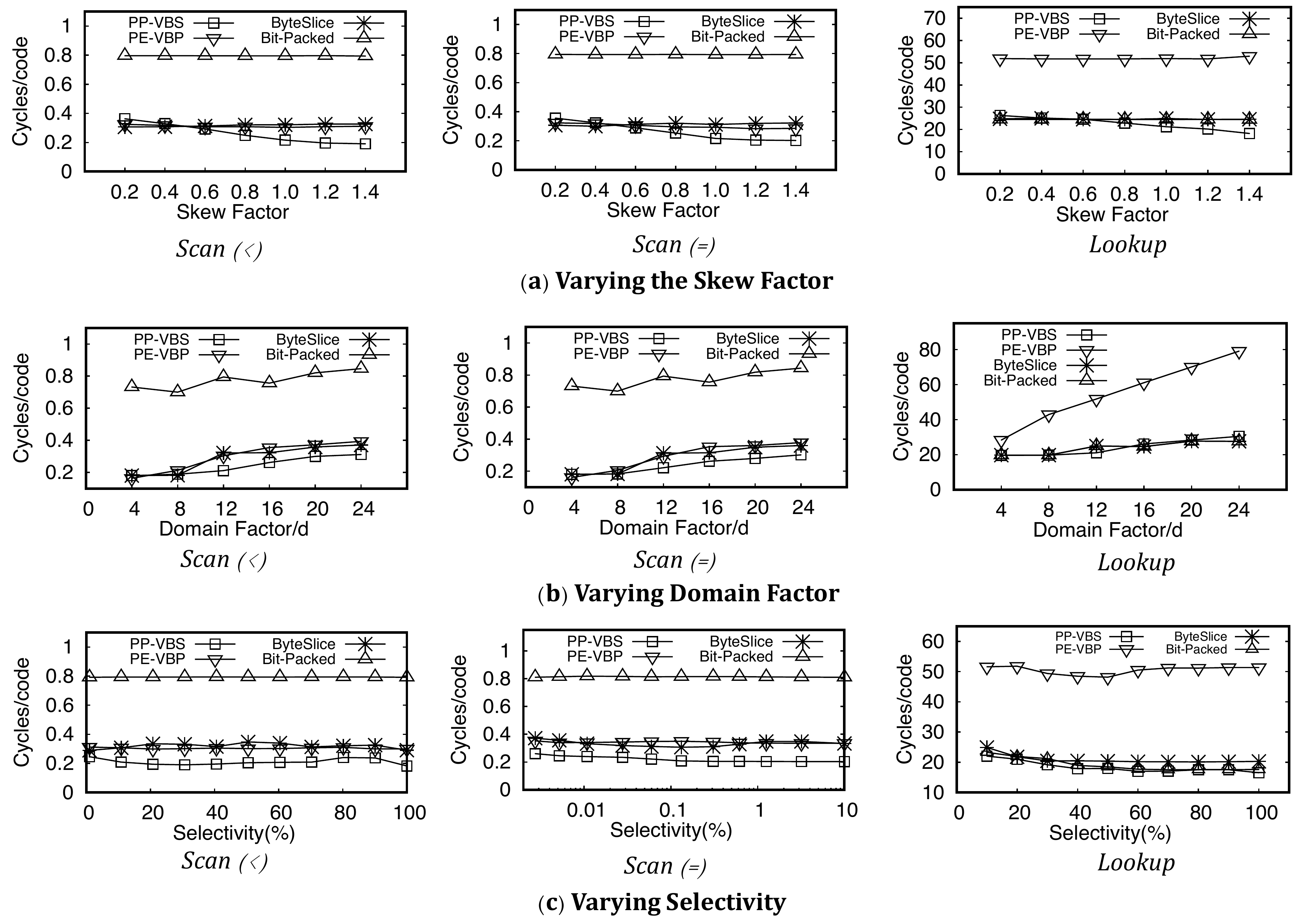}

\caption{Sensitivity Analysis}
\label{fig:whole-speedup}
\vspace{-0.2cm}
\end{figure*}

\begin{table*}[h]
%\scriptsize
\centering
\begin{tabular}{|c|c|c|c|c|c|}
%\begin{tabular}{|c|c|c||c|c|}
\hline
Dataset      & \# of Rows  & \# of Columns (ByteSlice:PP-VBS) & Encoding Time & Profiling \& Layout Selection Time & \# of Public Queries\\ \hline \hline
TAXI\_TRIP \cite{taxitrip}&  185,666k & 23 (6:17)& 180.21 min & 12.25 min & 3  \\ \hline
%INDIV16 \cite{fec}     &  20,557k & 21 (13:8) & 15.5 min& 0.59 min  & 2      \\ \hline
HEALTH \cite{health}  & \phantom{00}2,784k & 6 (4:2)  &  \phantom{00}1.12 min & \phantom{0}0.15 min & 1    \\ \hline
%LISTEN \cite{listen}   &  146,914,257 & 12 (7:5)& 135 min& 2.3 min           \\ \hline
NYC \cite{nyc} &  146,113k   &  19 (13:6) & 117.16 min & \phantom{0}8.03 min & 3\\ \hline
EDUCATION \cite{education} & \phantom{00}5,082k & 6 (3:3) & \phantom{00}1.98 min & \phantom{0}0.21 min & 1\\ \hline
FEC \cite{fec}     &  \phantom{0}20,557k & 35 (20:15) & \phantom{0}30.59 min& \phantom{0}2.58 min  & 2      \\ \hline
NPPES \cite{nppes} & \phantom{00}5,943k & 480 (353:127) & 123.88 min & 15.92 min & 2 \\ \hline
\end{tabular}
\caption{Real Datasets: Characteristics and Ingestion Time/Result}
\label{table:real_data}
\vspace{-0.4cm}
\end{table*}

Figure \ref{fig:whole-speedup}(a) reports the scan and lookup performance of different layouts when varying skew factor.
The results are averages from queries with 100 different selectivities. 
In this experiment, following \cite{paddedEncoding}, we set the domain size as $2^{12}$ (i.e., $d=12$). 
It is reasonable since in real-world datasets, the number of distinct values of columns (aka domain size) is mainly in the range of ($2^7$,$2^{16}$) \cite{li2013bitweaving,willhalm2009simd}.
%Figures \ref{fig:varyskew} show the performance of scan.
% We select 101 values for $G$ which control the selectivity from 0\% $\sim$ 100\% and report the average performance of these 101 queries under each techniques. 
As shown in Figure \ref{fig:whole-speedup}(a), ByteSlice is dominant on scan operation 
when the skew factor is less than 0.5;
PP-VBS starts to dominate when the skew factor increases.
PP-VBS, ByteSlice and PE-VBP achieve way better scan performance than Bit-Packed layout because of early stop.
% We also compare the lookup performance of different storage layouts. 
%
% By default, we set the selectivity of lookup is 10\%. 
% Unless stated otherwise, the lookups on ByteSlice and PE-VBP take a bitvector as input instead of a position list. 
%From the results, 
%we see that lookup is two orders more expensive than scan and it should not be ignored.
In terms of lookup, PP-VBS and ByteSlice perform as well as Bit-Packed and outperform
PE-VBP in all cases because they do not scatter the bits into so many different words as PE-VBP does.
%Lookups on PP-VBS and ByteSlice have comparable performance.
%Specifically, the lookup performance of PP-VBS improves along with the skew factor. 
Lookup performance on PP-VBS improves mildly when the data is getting more skewed because averagely each block of a column contains more short-length codes under higher skew. 
It decreases the number of bytes read from the memory, which benefits the memory-bound lookup operation.
% It decreases the number of bytes read from the memory, which benefits in the memory-bound lookup operation.

Figure \ref{fig:whole-speedup}(b) reports the scan and lookup performance of different layouts when varying the domain size (i.e., $2^4 \sim 2^{24}$).
We fix the skew factor as 1.0 in this experiment.
% Since in real world scenarios, and according to our experiments, the domain factors of dataset columns are rarely beyond 24. 
On skewed data, PP-VBS outperforms 
Bit-Packed, PE-VBP and ByteSlice under a wide range of domain size in terms of scan operation.
The cost of scan increases with the domain size because generally more bits of a code are retrieved from the memory for a predicate. %before early stop.
Lookup time also increases with the domain size 
because the average code length increases with domain size.

Figure \ref{fig:whole-speedup}(c) 
reports the scan and lookup performance of different layouts when varying selectivities. 
We fix skew factor as 1.0 and domain size as $2^{12}$ (i.e., $d = 12$) in this experiment. 
%PP-VBS dominates both PE-VBP and ByteSlice on scan for all selectivities.
PP-VBS dominates all the other storage layouts on scan for all selectivities.
Both PP-VBS and ByteSlice also have as excellent performance as Bit-Packed and outperform PE-VBP on lookup operation under all selectivites.

\begin{figure}[]
\centering 
\includegraphics[width=1.0\columnwidth]{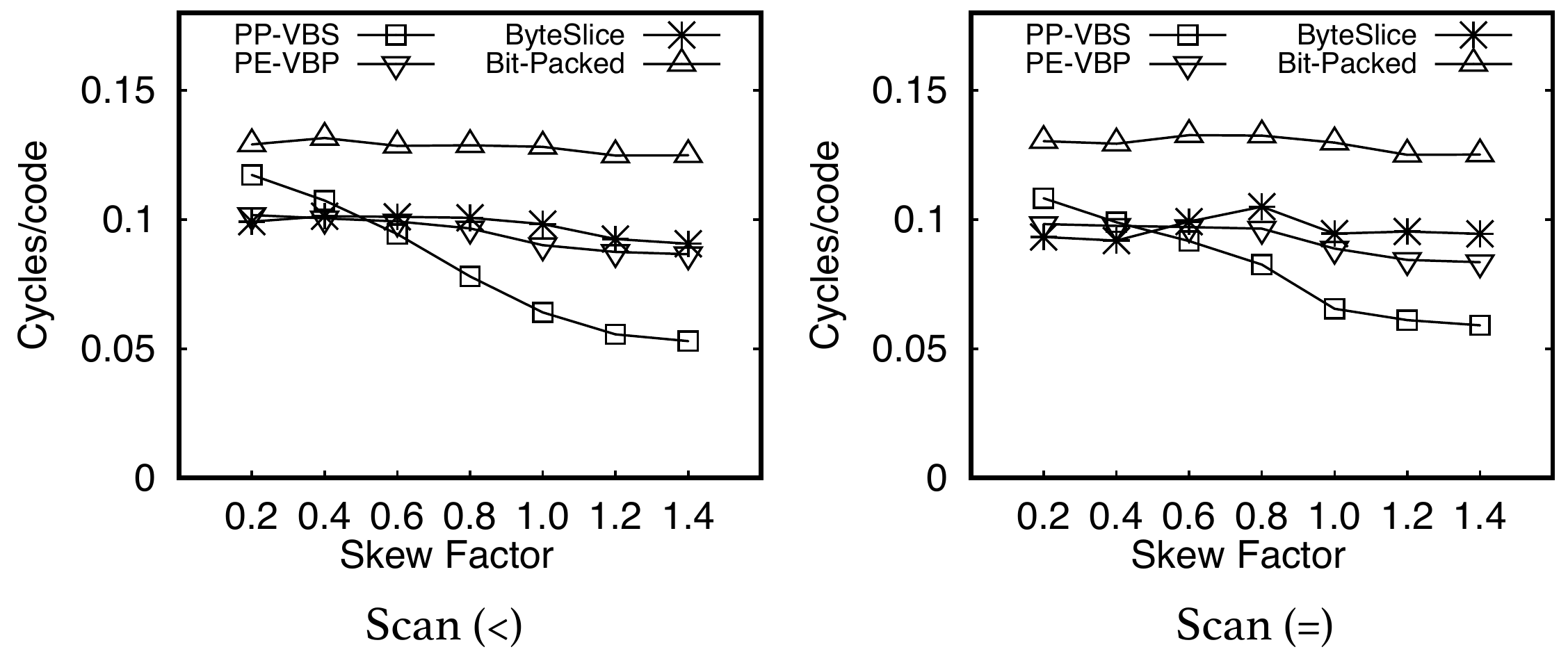}
\caption{Multi-threaded Scan (\#threads = 8)}
\label{fig:multi_threaded}
\vspace{-0.4cm}
\end{figure}

In addition, we evaluated our experiments of multiple threads. 
Figure ~\ref{fig:multi_threaded} reports the scan performance of different layouts when varying skew factor. 
In this experiment, we fixed the domain size as $2^{12}$ (i.e., $d$ = 12) and the number of threads as $8$, which means all CPU cores are used.
It is clear to see that we can draw similar conclusion from the multi-threaded experiments to that from the single-threaded counterpart.

we also have carried out experiments
of (i) varying the cardinalities of the columns,
(ii) using data generated by Gaussian distribution of different variances instead of using Zipf distribution,
(iii) using other operators (e.g., $\geq$). 
Since those experiments draw similar conclusions as the above,
we do not present them here for space interest.

\begin{figure}[]
\centering 
\includegraphics[width=1.0\columnwidth]{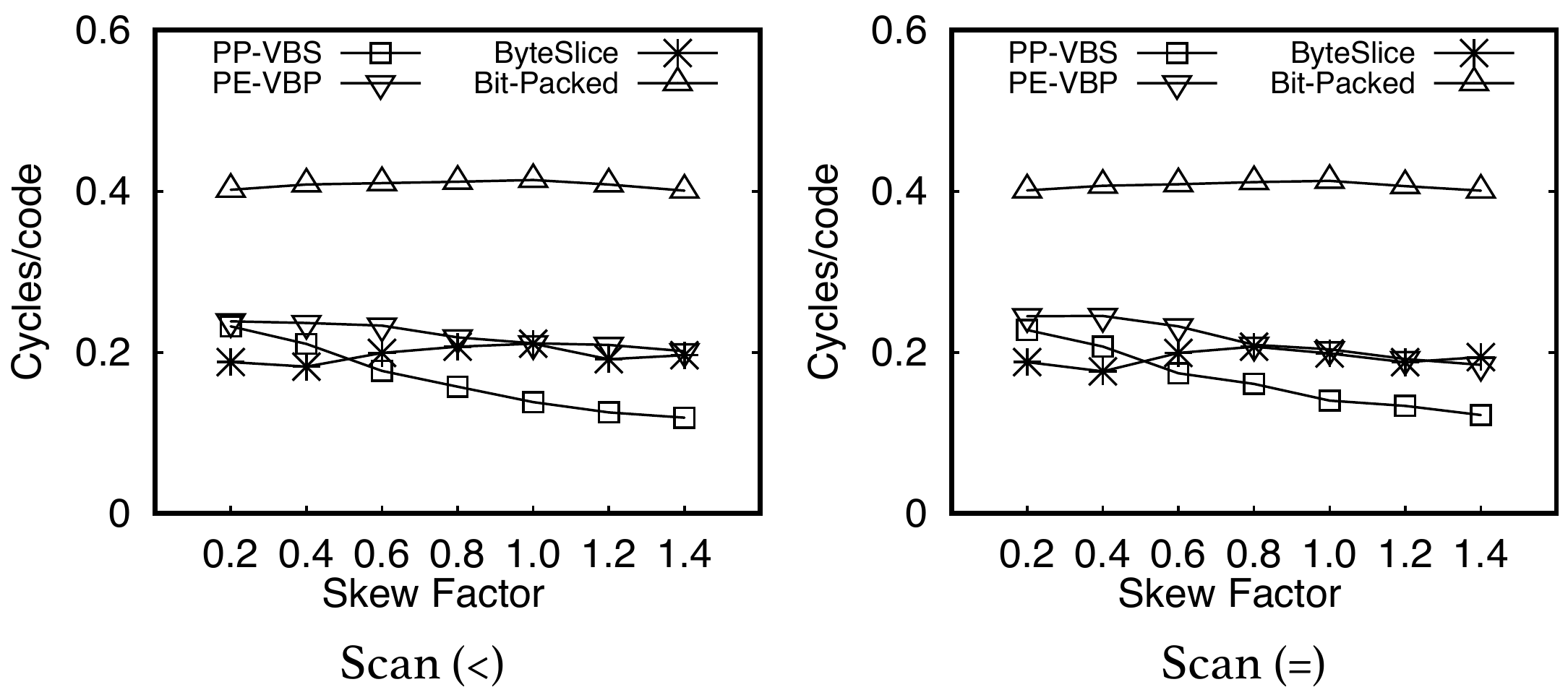}
\caption{Scan Performance with AVX-512}
\label{fig:avx512}
\vspace{-0.4cm}
\end{figure}

\textbf{AVX-512}: Our techniques are not specific to AVX2, they can be straightforwardly extended to AVX-512 model. 
We re-implemented PP-VBS, Bit-Packed, PE-VBP and ByteSlice with AVX-512 instructions (e.g., {\tt \_mm512\_cmpgt\_epu8\_mask()} \footnote{This instruction compares 64 pairs of unsigned 8-bit integers in 512-bit SIMD registers for greater-than.}) and ran the experiments on a 2.2GHz 10-core Intel Xeon Silver 4114 processor. 
Figure \ref{fig:avx512} reports the scan performance of different layouts when varying skew factor using single core. 
Domain size is fixed as $2^{12}$ in this experiment. 
We can obtain the similar conclusion that ByteSlice is dominant on scan operation when the skew factor is less than 0.5; PP-VBS starts to dominate when the skew factor increases.

\begin{figure}[]
\centering 
\includegraphics[width=1.0\columnwidth]{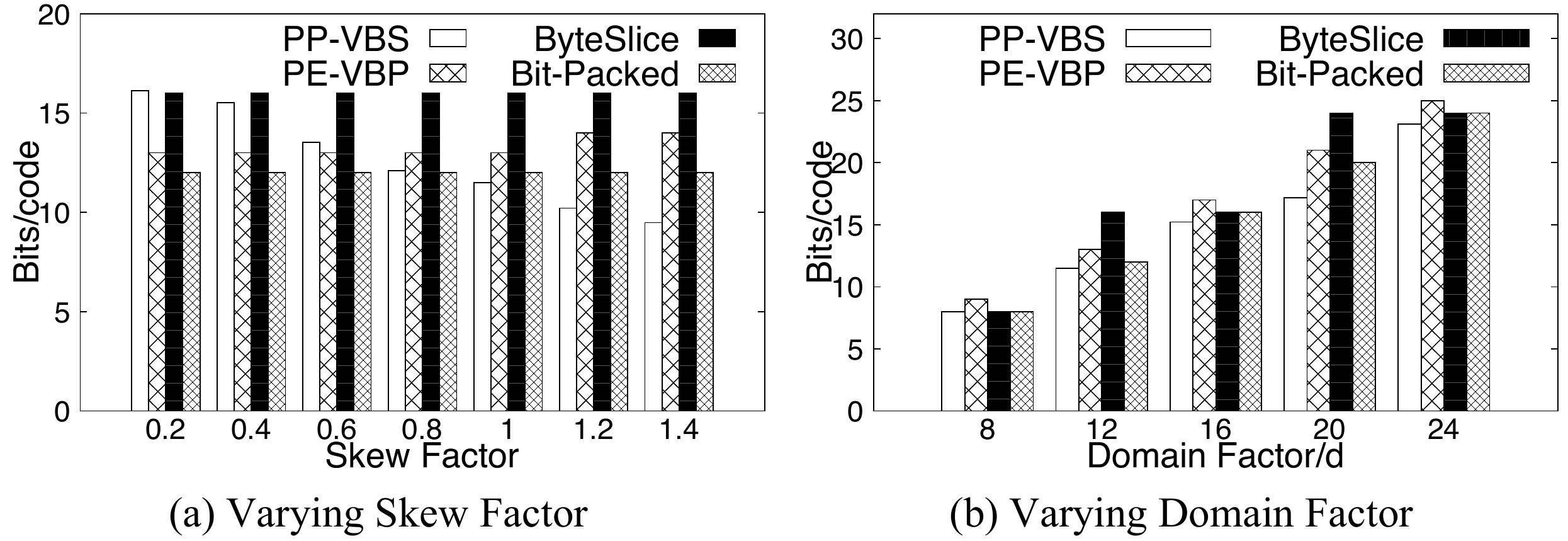}
\caption{Memory Usage of Different Layouts}
\label{fig:memory_usage}
\vspace{-0.4cm}
\end{figure}

\textbf{Memory Usage}: Figure~\ref{fig:memory_usage} reports the average number of bits used per code under different storage layouts. 
In Figure~\ref{fig:memory_usage}(a), we fixed the domain size as $2^{12}$ and varied the skew factor. 
Not surprisingly, the average number of bits per code under PP-VBS reduces as the skew factor increases and PP-VBS starts to dominate other storage layouts when the skew factor is larger than 0.8. 
Notably, we also counted the memory usage of the bitmasks $\mathcal{M}$.
The memory usage of PE-VBP increases along with the skew factor since the maximum code length increases with the skew factor.
For ByteSlice, even though the code length is 12-bit, one code must take up 16 bits under ByteSlice because of padding zeros, which will waste memory. 
However, we focus on scan and lookup performance in this paper and select ByteSlice as a candidate since it dominates scan and lookup performance on uniform to lightly skewed data columns.
As opposite to Figure~\ref{fig:memory_usage}(a), we fixed the skew factor as $1$ and varied the domain factor in Figure~\ref{fig:memory_usage}(b). 
PP-VBS uses the least memory when columns are skewed under varying domain factor.
\subsection{Real Data Evaluation}
This set of experiments aims to evaluate ByteStore
as a whole and compare it with homogenous storage engines that use only Bit-Packed, use only ByteSlice or only PE-VBP.

The experiments are done using 6 real datasets downloaded from Google BigQuery
in 2019 May \cite{bigquery}.
Table \ref{table:real_data} shows their details
as well as the results of dataset ingestion by ByteStore's column layout advisor.
It clearly shows that different columns in
a dataset require different storage layouts.
The column layout advisor does not use much time 
to do the profiling and layout selection. 
The offline data ingestion time is mainly spent on encoding the columns,
but half of that time is indispensable as the column has to be encoded in one of the two storage layouts anyway. 
To focus only on scans and lookups, we follow \cite{li2014widetable} to materialize the joins and execute the selection-projection components of the queries.
% For each datasets, BigQuery provides several (1 to 3) public queries. 
% The last column of Table \ref{table:real_data} lists the queries we evaluate.
Same as \cite{byteslice}, we discard the queries that have no selection clause and queries that involve string similarity comparison {\tt LIKE}. 
In interest of space, the details of these queries are not presented, please refer to the main page of each dataset for more information.

%\begin{figure}
%\centering
%\subfigure[Scan Speed-up over PE-VBP]{\label{fig:scanspeedup}
%\includegraphics[width=1.0\columnwidth]{figures/scan_1.pdf}}
%\subfigure[Lookup Speed-up over PE-VBP]{\label{fig:lookupspeedup}
%\includegraphics[width=1.0\columnwidth]{figures/lookup_1.pdf}}
%\subfigure[Query Speed-up over PE-VBP]{\label{fig:speedup}
%\includegraphics[width=1.0\columnwidth]{figures/queries_1.pdf}}
%\caption{\centering{Performance of Different Layouts on Five Real Datasets: TAXITRIP (Queries T*), NYC (Queries N*), EDUCATION (Queries E*), FEC (Queries F*), NPPES (Queries P*)}}
%\label{fig:real_performance}
%\vspace{-0.6cm}
%\end{figure}

\begin{figure}[]
\centering 
\subfigure[Execution Time Breakdown for Queries. Y-axis reports cycles per tuple]
{
\includegraphics[width=1.0\columnwidth]{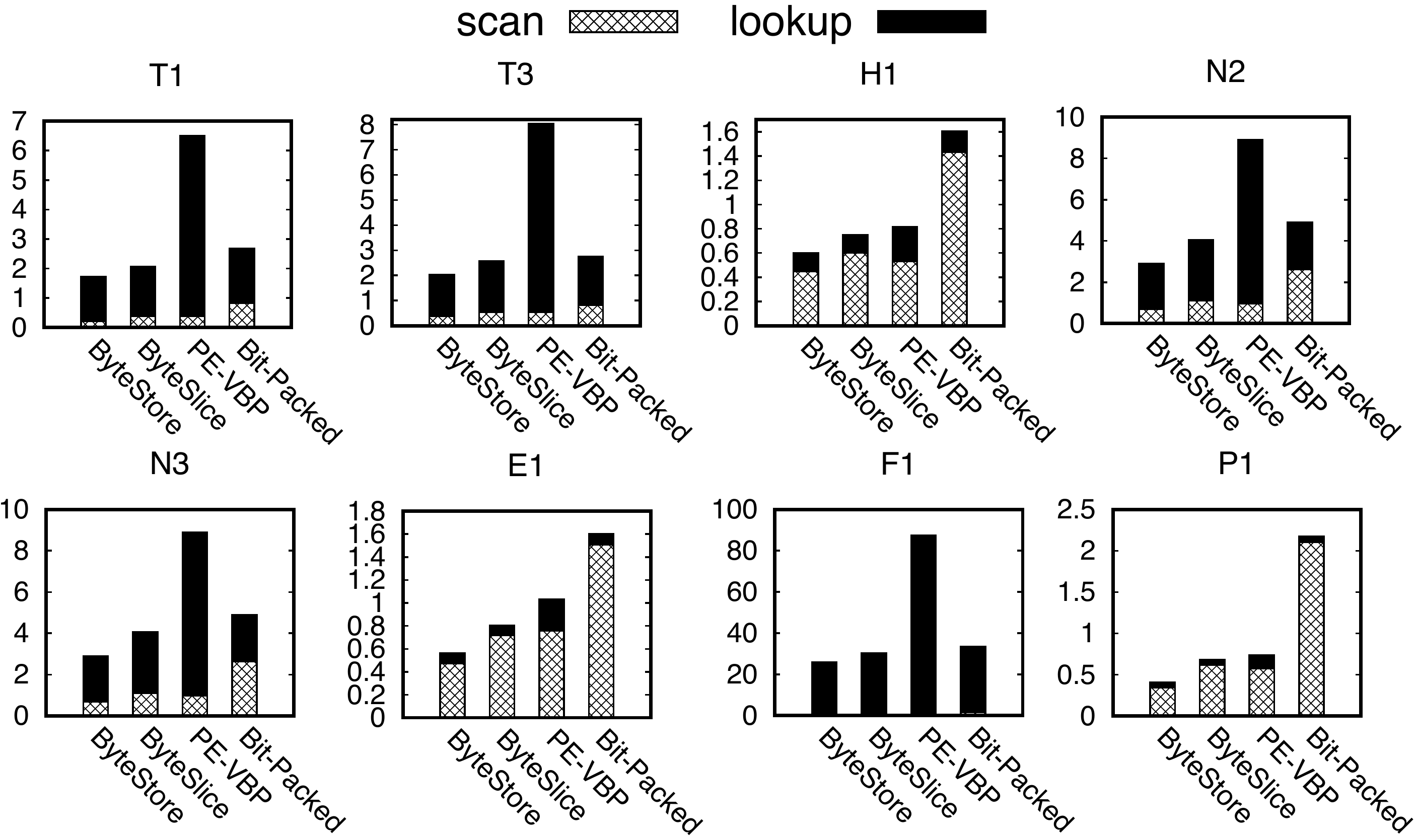}
}
\subfigure[Query Speed-up Over PE-VBP]
{
\includegraphics[width=1.0\columnwidth]{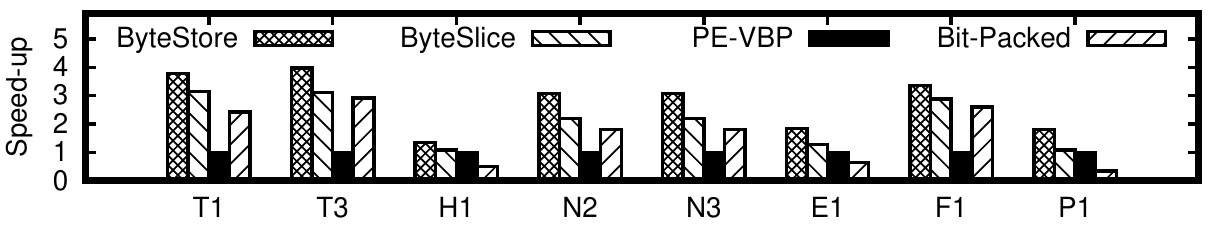}
}
\caption{
Performance Comparison of Different Layouts on Six Real Datasets: TAXITRIP(Queries T*), HEALTH (Queries H*), NYC(Queries N*), EDUCATION (Queries E*), FEC (Queries F*), NPPES (Queries P*)
%\centering{Performance Comparison of Different Layouts on Six Real Datasets: TAXITRIP(Queries T*), HEALTH (Queries H*), NYC(Queries N*), EDUCATION (Queries E*), FEC (Queries F*), NPPES (Queries P*)}
}
\label{fig:breakdown}
\vspace{-0.4cm}
\end{figure}

Figure \ref{fig:breakdown} compares ByteStore with using ByteSlice only, using PE-VBP only, and using Bit-Packed only across different queries and different datasets.
%The queries on that six datasets are also all real public queries available from Google BigQuery. 
%Queries that have no selection clauses are discarded.
%The results are presented as (a)the time breakdown of all queries and (b)speed-up over using the PE-VBP only.
In Figure~\ref{fig:breakdown}(a), we report the execution time breakdown of all queries. 
The run time of each query is dissected into scan cost and lookup cost. 
The reported numbers have been normalized on a per tuple basis. 
We could see there are both scan-dominant (e.g., E1 and P1) and lookup-dominant (e.g., F1) queries. 
%The results show that ByteStore dominates scan, lookup and the whole query performance.
We can see that ByteStore outperforms all homogeneous storage schemes.
Figure~\ref{fig:breakdown}(b) reports the query speed-up over PE-VBP only. 
Overall, ByteStore brings up to 4.0$\times$, 1.7$\times$ and 5.2$\times$ speedup to query performance
when comparing with using PE-VBP only, ByteSlice only and Bit-Packed only respectively.

\section{Related Work} \label{sec:related}

%{\bf need to add a paragraph about hybrid data engine}
Since the beginning of 21st century, there have been several hybrid data engines developed for HTAP workloads.
Hybrid storage engines mainly focus on mixing the row and column representations 
\cite{hyrise2010,kemper2011hyper, featureMirrors,LeeMMFSPKG13, bridgingPalvo,datablock,oracle_database}.
Among them, fractured mirrors \cite{featureMirrors} advocates the maintenance of both NSM (row-oriented) and DSM (column-oriented) physical representations of the database simultaneously.
HYRISE \cite{hyrise2010} automatically partitions the tables into variable-length vertical segments based on how the attributes of each table are co-accessed by the queries. 
%HYRISE provides better cache utilization than PAX when scanning both narrow and wide projection. 
Based on the work from HYRISE, SAP developed HANA 
that starts out with a NSM layout, and then migrate to a compressed DSM storage manager \cite{LeeMMFSPKG13}.
Another DBMS that supports dual NSM/DSM storage like HANA is MemSQL \cite{memsql}. 
Like HANA, these storage layouts are managed by separate runtime components. 
Unlike HANA, MemSQL exposes different layout options to the application.
None of the above focus on pure column store. 
In this work, we focus on that and study the use of different storage layouts for different columns.

%{\bf need to add a paragraph about experiment-driven advisor, see Shivnath Babu papers}
The experiment-driven approach \cite{experimentDriven} has been 
used to tune database systems \cite{itunes}, batch systems \cite{Badu2010}, and machine learning systems \cite{tftuner}.
% online index building \cite{onlineIndex}, database migration  \cite{zephyr2011}, and defragmentation \cite{}.
% This line of work helps the migration from one host to another or from one configuration to another. 
In this paper, we also use the experiment-driven approach to design our column-layout-advisor.  
The advantage of experiment-driven approach is that its results are highly accurate 
at the cost of relatively longer tuning time.  Nonetheless, tuning is an offline process and the tuning time is not a crucial factor.

Encoding techniques have been extensively used for main memory analytical databases in both the research community \cite{abadi2006integrating,pe_10,krueger2011fast,li2013bitweaving}, and in the industry, e.g., SAP HANA~\cite{farber2012sap}.
Works listed above all use the fixed length encoding and do not leverage column skew.
IBM Blink \cite{blink} and its commercial successor IBM DB2 BLU \cite{DB2BLU} employ a proprietary encoding technique called frequency partitioning. 
%The first step in compressing a column is to analyze the column distribution and accordingly determine the way to partition the column. 
%Next, a separate dictionary is created for each column partition. 
%The code length within a partition is fixed, while different partitions can have different code lengths.
Based on the frequency of data, a column is divided into multiple partitions, each of
which uses an independent fixed-length encoding. 
The code lengths of different column partitions are different.
%The partitions containing frequent column values use shorter codes, while the partitions containing infrequent column values use longer codes.
%The first step in compressing a table is to analyze the column distributions and accordingly determine the best way to partition the table. Next, a separate dictionary is created of values in each partition. Partitions containing frequent column values use shorter codes. 
%{\bf (Eric still doesn't understand the statement below)}
Thus, that technique can be viewed as a hybrid between the fixed-length encoding and variable-length encoding. By contrast, PPE is a pure variable
length encoding scheme and it should work with a distinct storage layout VBS in tandem.

%Thus, this technique can be viewed as a hybrid between the fixed-length encoding and variable-length encoding. While PPE is a pure variable length encoding scheme and it should work with a distinct storage layout VBS in tandem.
%So, a value may have variable lengths across the partitions.
%Nonetheless, {\bf try to write something bad about this approach when comparing with us}.

%To further improve the scan performance of OLAP workloads, lightweight indexes are popular techniques in the literature \cite{zone_maps,oracle_database,imprint,columnSketches,datablock}. 
Lightweight indexes are techniques that skip data processing by using summary statistics over the base column. 
Such techniques include Zone Maps \cite{zone_maps}, Column Imprints \cite{imprint}, Feature Based Data Skipping \cite{feature_data_skipping}, Column Sketches \cite{columnSketches} and BinDex \cite{bindex}. 
For example, as a widely used technique, Zone Maps partition a column into zones and record the metadata of each zone, such as \emph{min} and \emph{max}. 
With data partitioning, the approaches skip zones where all values in the zone satisfy or not satisfy the predicate.
% Lightweight indexes are widely used as an add-on in column stores such as Oracle Database In-memory \cite{oracle_database}, Data Blocks \cite{datablock}, MonetDB \cite{imprint,columnSketches,bindex} and etc.  
%Oracle database in-memory \cite{oracle_database} splits each column into sections and maintains an in-memory storage index, which keeps track of minimum and maximum values for each column section.
%This allows scan operation to quickly skip processing the column sections which do not contain matching values. 
%Similarly, Data Blocks \cite{datablock} keeps a light weight index named positional SMA (PSMA). 
%Column Imprints \cite{imprint} and Column Sketches \cite{columnSketches} are indexing techniques which can be applied to various column stores (e.g., MonetDB). 
%PSMA, Column Imprints and Column Sketches stores small metadata for groups of data of a base column and enables data skipping to improve the scan performance.
These lightweight indexing techniques are complementary with storage layouts
and can be used together.

\section{Conclusion} \label{sec:conclusion}

Choosing the optimal layout for individual columns in scan-based
OLAP systems is non-trivial because it must balance between scan
and lookup performance and account for the column data characteristics.
In this paper, 
we first presented a new layout, PP-VBS, that achieves both fast scan and 
fast lookup on skewed data. We then described ByteStore,
a hybrid column store
using an experiment-driven approach to select the best column 
layout for each individual column. Experiments on real and synthetic datasets and workloads
%in addition to TPC-H 
show that our hybrid column store significantly
outperforms previous works in terms of end-to-end query performance.

% Extending our techniques to future SIMD architectures with wider registers should be straightforward. 
% If instruction sets with novel semantics emerge (e.g., more sophisticated bit manipulations), 
% it will be interesting to
% revisit PP-VBS for further optimization. 
% Implementing PP-VBS on
% FPGAs~\cite{FPGABingsheng} and approximate hardware~\cite{APPBingsheng}
% is also a promising direction.
%There are other data structures to accelerate scans
%such as column imprint \cite{imprint} and column sketches \cite{columnSketches}.

% \section{Acknowledgment}
% This work is supported by Hong Kong General Research Fund (14200817, 15200715, 15204116), Hong Kong AoE/P-404/18, Innovation and Technology Fund ITS/310/18.

\bibliographystyle{IEEEtran}
\bibliography{IEEEabrv,0_references}

\end{document}